\documentclass[12pt,reqno] {amsart}

 \usepackage{amsmath, wasysym}
 \usepackage{bbm, dsfont}
 \usepackage{amssymb}
\usepackage{graphicx}%
\usepackage{hyperref}
\hypersetup{hidelinks}
\usepackage{cleveref}
\usepackage{comment}
 \usepackage{mathrsfs}
 \usepackage{float}
 \usepackage{tikz, tikz-cd}
  \usetikzlibrary{arrows.meta, knots, calc}
 \usepackage{bm}
 \usepackage{mathtools} 
\usepackage{extarrows} 
\usepackage{bbm}

 \usepackage{enumerate}
  \usepackage{mathdots}
\usetikzlibrary{decorations.markings}
\usetikzlibrary{decorations.pathreplacing}

\theoremstyle{plain}
\newtheorem{theorem}{Theorem}[section]

\newtheorem{corollary}[theorem]{Corollary}
\newtheorem{lemma}[theorem]{Lemma}
\newtheorem{proposition}[theorem]{Proposition}
\newtheorem*{ac}{Acknowledgements}

\newtheorem{remark}[theorem]{Remark}

\newtheorem{example}[theorem]{Example}

\newcommand{\bC}{\mathbb{C}}
\newcommand{\bN}{\mathbb{N}}
\newcommand{\bR}{\mathbb{R}}
\newcommand{\cB}{\mathcal{B}}
\newcommand{\cL}{\mathcal{L}}
\newcommand{\cD}{\mathcal{D}}
\newcommand{\cF}{\mathcal{F}}

\newcommand{\cH}{\mathcal{H}}

\newcommand{\bE}{\mathbb{E}}

\newcommand{\Tr}{\mathrm{Tr}}
\newcommand{\Id}{\mathrm{Id}}

\newcommand{\bfM}{\mathbf{M}}

\newcommand{\KMS}{\mathrm{KMS}}

\newcommand{\Ent}{\mathrm{Ent}}
\newcommand{\diag}{\mathrm{diag}}

\allowdisplaybreaks[2]
\renewcommand{\Re}{\operatorname{Re}}
\renewcommand{\Im}{\operatorname{Im}}

\newcommand{\bfa}{\mathbf{a}}
\newcommand{\bfJ}{\mathbf{J}}
\newcommand{\bfz}{\mathbf{z}}
\newcommand{\bfZ}{\mathbf{Z}}
\newcommand{\bfw}{\mathbf{w}}
\newcommand{\bfm}{\mathbf{m}}
\newcommand{\bfC}{\mathbf{C}}
\newcommand{\bfF}{\mathbf{F}}
\newcommand{\bfzeta}{\boldsymbol{\zeta}}
\newcommand{\hc}{h_\rho}
\newcommand{\thetaK}{\theta_{\KMS}}

\newcommand{\norm}[1]{\left\lVert #1\right\rVert}
\newcommand{\np}[2]{\norm{#1}_{#2,\rho}}

\newcommand{\arsinh}{\operatorname{arsinh}}

\title[Entropy Contraction and Hypercontractivity]{Entropy Contraction and Hypercontractivity for Gaussian Quantum Markov Semigroups}
\author{Zhengwei Liu}
\address{Zhengwei Liu, Tsinghua University, Beijing}
\email{liuzhengwei@tsinghua.edu.cn}

\author{Jincheng Wan}
\address{Jincheng Wan, Tsinghua University, Beijing}
\email{wanjc23@mails.tsinghua.edu.cn }

\author{Jinsong Wu}
\address{Jinsong Wu, Beijing Institute of Mathematical Sciences and Applications, Beijing, 101408, China}
\email{wjs@bimsa.cn}

\date{}

\begin{document}

\begin{abstract}
In this paper, we study finite-mode Gaussian quantum Markov semigroups with a faithful invariant Gaussian state. 
We present an algebraic characterization of the complete modified logarithmic Sobolev inequality relative to the fixed-point algebra, and we obtain the corresponding optimal constant under the assumption that the drift evolution $e^{t\bfZ}$ converges as $t\to\infty$.
We establish hypercontractivity for such quantum Markov semigroups with Hurwitz drift, in the sense that every fixed pair $1<p<q<\infty$ is attained at sufficiently large times. 
We find that the usual reverse hypercontractive curves with a positive rate from time zero can fail in general, but every fixed pair \(1/2 < p < q < 1\) is attained at sufficiently large times under the Hurwitz assumption.
We also systematically investigate the \(p\)-logarithmic Sobolev inequality and show that the \(p\)-logarithmic Sobolev constant can be negative for \(0 < p < 1/2\).
\end{abstract}

\maketitle

\section{Introduction}

Gaussian quantum Markov semigroups \cite{KP04,CEGH08,DP18,GP26} are powerful tools in studying the evolution of open bosonic systems, which are quantum counterparts of the classical Ornstein--Uhlenbeck semigroups.
Fundamental properties of Gaussian quantum Markov semigroups, such as irreducibility, spectral gaps, complete modified logarithmic Sobolev inequalities, hypercontractivity, reverse hypercontractivity, and logarithmic Sobolev inequalities, are essential for understanding the evolution of open bosonic systems.
The irreducibility of Gaussian quantum Markov semigroups has been investigated by Fagnola and Girotti \cite{FG26}.
The spectral gaps of Gaussian quantum Markov semigroups have been systematically investigated by Fagnola, Poletti, Sasso, and Umanità, and by Li \cite{FPSU25,Li26}.

The complete modified logarithmic Sobolev inequality was investigated by Brannan, Gao, and Junge \cite{BGJ22,BGJ23}.
In \cite{GaoRou22}, Gao and Rouzé showed that GNS-symmetric quantum Markov semigroups with a faithful invariant state on finite-dimensional von Neumann algebras satisfy the complete modified logarithmic Sobolev inequality relative to the fixed-point algebra.
The authors of \cite{LWW26} pointed out that a KMS-symmetric quantum Markov semigroup might not satisfy the complete modified logarithmic Sobolev inequality.
Recently, Gao and Guo \cite{GaoGuo26a,GaoGuo26b} presented a complete characterization of complete modified logarithmic Sobolev inequalities for finite-dimensional quantum Markov semigroups admitting a faithful asymptotic conditional expectation, with improved constants.
In this paper, we present an algebraic characterization of complete modified logarithmic Sobolev inequalities relative to the fixed-point algebra for finite-mode Gaussian quantum Markov semigroups, assuming that $\lim_{t\to\infty}e^{t\bfZ}$ exists.
Under this assumption, we also determine the optimal complete entropy contraction coefficient at every fixed time.
Moreover, we show that positive relative CMLSI forces the drift evolution to converge.
Inspired by the work of Gao and Guo, we also obtain improved complete modified logarithmic Sobolev constants for Gaussian quantum Markov semigroups.

The general framework for noncommutative hypercontractivity was obtained by Olkiewicz and Zegarlinski \cite{OZ99}, who established the equivalence between hypercontractivity and logarithmic Sobolev inequalities.
In \cite{BDR20}, Beigi, Datta, and Rouzé systematically investigated quantum hypercontractivity for quantum Markov semigroups.
They showed that primitive GNS-symmetric quantum Markov semigroups on finite-dimensional matrix algebras possess hypercontractivity via the quantum Stroock--Varopoulos inequality.
Recently, Gao and Wang \cite{GaoWang26} and Wan and Wu \cite{WanWu26} proved hypercontractivity for primitive KMS-symmetric quantum Markov semigroups on finite-dimensional matrix algebras.
Carbone and Sasso \cite{CS08} proved hypercontractivity for quantum Ornstein--Uhlenbeck semigroups.
Beigi and Rahimi-Keshari \cite{BRK25} showed the logarithmic Sobolev inequality for one-mode quantum Ornstein--Uhlenbeck semigroups.
In \cite{SXZ26}, Sun, Xu, and Zhang proved hypercontractivity for a family of quantum Ornstein--Uhlenbeck semigroups.
In this paper, we show that finite-mode Gaussian quantum Markov semigroups with a faithful invariant centered Gaussian state and Hurwitz drift attain every fixed forward exponent pair $1<p<q<\infty$ at sufficiently large times.
Under the stronger condition $\gamma_{\KMS}>0$, we obtain explicit forward exponent curves with a positive rate from time zero.
This condition holds in the KMS-symmetric case with Hurwitz drift, by Proposition~\ref{prop:hurwitzpos}.

Quantum reverse hypercontractivity~\cite{CKMT15} was systematically studied by Beigi, Datta, and Rouzé \cite{BDR20}.
They showed that primitive GNS-symmetric quantum Markov semigroups on finite-dimensional matrix algebras possess quantum reverse hypercontractivity.
However, the third author \cite{Wu26} showed that a positive uniform reverse hypercontractivity rate on $(0,1)$ need not exist for primitive KMS-symmetric quantum Markov semigroups on finite-dimensional matrix algebras, but reverse hypercontractivity does hold in the restricted range $1/2<p<q<1$.
We show that this phenomenon also appears for Gaussian quantum Markov semigroups.
We also present finite-mode examples for which the usual reverse exponent curves cannot have a strictly negative initial derivative, and a finite-mode example with negative $p$-logarithmic Sobolev constants for $0<p<1/2$.
Under the Hurwitz assumption, every fixed pair $1/2<p<q<1$ is nevertheless attained at a sufficiently large time.
If $\gamma_{\KMS}>0$, the reverse exponent curves in this range also have a positive rate from time zero.


The paper is organized as follows.
In Section 2, we recall the basic notations for Gaussian states and Gaussian quantum Markov semigroups, and revisit GNS and KMS symmetry.
In Section 3, we characterize the relative complete modified logarithmic Sobolev inequalities under the fixed-sector and relaxing-sector decomposition.
In Section 4, we show that finite-mode Gaussian quantum Markov semigroups with a faithful invariant centered Gaussian state and Hurwitz drift attain every fixed forward exponent pair at sufficiently large times.
In Section 5, we show reverse hypercontractivity for these semigroups on a restricted range.
In Section 6, we present examples illustrating the hypercontractivity and reverse hypercontractivity results and the role of their assumptions.
In Section 7, we investigate the $p$-logarithmic Sobolev inequalities and construct a finite-mode counterexample for $0<p<1/2$.

\begin{ac}
Zhengwei Liu was supported by Beijing Natural Science Foundation Key Program (Grant No. Z220002) and Beijing Natural Science Foundation (Grant No. Z221100002722017).
Jinsong Wu was supported by grants from Beijing Institute of Mathematical Sciences and Applications.
Jinsong Wu was supported by NSFC (Grant no. 12371124). 
\end{ac}

\section{Preliminaries}

In this section, we shall recall the notations for Gaussian states, Gaussian quantum Markov semigroups etc.

\subsection{Gaussian States}
Suppose that $\displaystyle \cF_s(\bC^d)= \bigoplus_{n=0}^{\infty}(\mathbb C^d)^{\otimes_s n}$ is the symmetric Fock space with canonical orthonormal basis 
\begin{align*}
    e_{j_1, \ldots, j_d}, \quad j_1, \ldots, j_d\in \bN_0.
\end{align*}
Here $\bN=\{1,2,\ldots\}$ and $\bN_0=\{0,1,2,\ldots\}$.
The exponential vector $e_{z}$ is defined as
\begin{align*}
     e_z=\bigoplus_{n\geq0}\frac{z^{\otimes n}}{\sqrt{n!}},\qquad z\in\mathbb C^d.
\end{align*}
The annihilation $\bfa_j$  and creation operators $\bfa_{j}^\dag$ are given by
\begin{align*}
    \bfa_k e_{j_1, \ldots, j_d}=& \sqrt{j_k} e_{j_1, \ldots, j_k-1, \ldots, j_d}, \\
    \bfa_k^\dag e_{j_1, \ldots, j_d}=& \sqrt{j_k+1} e_{j_1, \ldots, j_k+1, \ldots, j_d}.
\end{align*}
We have that $[\bfa_k, \bfa_\ell^\dag]=\delta_{k\ell}\mathbbm{1}$.
Let $\cD$ be the finite linear span of the basis vectors in $\cF_s(\bC^d)$.

For each $z=(z_1, \ldots, z_d)^{\mathsf{T}}\in \bC^{d}$, the Weyl operator $W(z)$ is 
\begin{align}\label{eq:weyl}
  W(z) =\exp\left( \sum_{j=1}^d z_j\bfa_j^\dag - \overline{z_j} \bfa_j\right).  
\end{align}
We have that 
\begin{align}\label{eq:Weyl multiplication convention}
 W(z)W(w)  = e^{-i \Im \langle z, w\rangle} W(z+w). 
\end{align}
Denote by $\bfJ= \begin{pmatrix} 
    0 & I_d\\ -I_d & 0
    \end{pmatrix}$.
Suppose $z=x+iy=(x_1, \ldots, x_d)^{\mathsf{T}} +i (y_1, \ldots, y_d)^{\mathsf{T}}$ and let $\bfz=\begin{pmatrix}
    x\\ y
\end{pmatrix}$.
Then we have that $\Im \langle z, w\rangle =\bfz^{\mathsf{T}} \bfJ \bfw$.

For $j=1, \ldots, d$, let 
\begin{align}
 Q_j=\frac{\bfa_j+\bfa_j^\dagger}{\sqrt2},\qquad
 P_j=\frac{\bfa_j-\bfa_j^\dagger}{i\sqrt2}.
\end{align}
Then the field operator is
\begin{equation}\label{eq:fields}
 R(z)=\sqrt2\sum_{j=1}^d\bigl((\Im z_j)Q_j-(\Re z_j)P_j\bigr).
\end{equation}
With this convention, $W(z)=e^{iR(z)}$.
We then see that $[Q_j,P_k]=i\delta_{jk}$ and
\begin{align}
 \left[R(z),R(w)\right]=2i\,\bfz^{\mathsf{T}}\bfJ\bfw.
\end{align}

A normal Gaussian state $\rho$ on $\cB(\cF_s(\bC^d))$ is defined by
\begin{equation}\label{eq:char}
 \rho\left(W(\bfz)\right)
 =\exp\left(-\frac12\bfz^{\mathsf{T}}\Sigma\bfz+i \bfm^{\mathsf{T}}\bfz\right),
\end{equation}
where $\Sigma$ is real symmetric in $M_{2d}(\bR)$ and $\Sigma+i\bfJ\geq 0$. The matrix $\Sigma$ is called the covariance matrix of $\rho$ and the vector $\bfm\in \bR^{2d}$ is the mean vector of $\rho$. 
A Gaussian state $\rho$ is centered if $\bfm=0$.
Note that the state $\rho$ is faithful exactly when
\begin{equation}\label{eq:faithful}
 \Sigma+i\bfJ>0.
\end{equation}
Equivalently, all symplectic eigenvalues of $\Sigma$ exceed one by Williamson's theorem. 
By Equation \eqref{eq:weyl} for the Weyl operator, the vacuum covariance is $\Sigma=I_{2d}$.

By taking a real symplectic matrix $\mathbf M$ such that
\begin{equation}\label{eq:Williamson}
 \Sigma=\mathbf M^{\mathsf{T}} D_s\mathbf M,\qquad
 D_s=\diag(s_1,\ldots,s_d,s_1,\ldots,s_d),\quad s_j>1,
\end{equation}
and defining
\begin{equation}\label{eq:beta}
 \beta_j=\log\frac{s_j+1}{s_j-1},\qquad s_j=\coth(\beta_j/2),
\end{equation}
we have that the centered state up to a Gaussian unitary is
\begin{equation}\label{eq:Williamson decomposition}
 \bigotimes_{j=1}^d\rho_{\beta_j},\qquad
 \rho_\beta=(1-e^{-\beta})\exp\left(-\beta \bfa^\dag\bfa\right).
\end{equation}
Here $\rho_\beta$ is the one-mode thermal density and $\bfa$ is the one-mode annihilation operator.
All weighted inequalities below are invariant under simultaneous unitary conjugation of the state and the dynamics. Displacements therefore allow us to remove the mean of an invariant Gaussian state.



\subsection{Gaussian Quantum Markov Semigroups}

Suppose that the Gaussian state $\rho$ evolves by the means and covariances:
\begin{equation}\label{eq:covariance}
 \dot \bfm_t=\bfZ^{\mathsf{T}} \bfm_t+\bfzeta,\qquad
 \dot\Sigma_t=\bfZ^{\mathsf{T}}\Sigma_t+\Sigma_t\bfZ+\bfC,
\end{equation}
where $\bfzeta\in \bR^{2d}$, $\bfm_0=\bfm$ and $\Sigma_0=\Sigma$.
Let $\displaystyle \bfC_t=\int_0^t e^{s\bf Z^{\mathsf{T}}}{\bf C} e^{s\bf Z}\,ds$, and $\displaystyle \bfzeta_t=\int_0^t e^{s\bfZ^{\mathsf{T}}}\bfzeta\,ds$.
Then the Gaussian quantum Markov semigroup $\{\Phi_t\}_{t\geq 0}$ on $\cB(\cF_s(\bC^d))$ is given by
\begin{equation}\label{eq:Tweyl}
 \Phi_t\left(W(\bfz)\right)= e^{i\langle \bfzeta_t, \bfz\rangle}
 \exp\left(-\frac12\bfz^{\mathsf{T}} \bfC_t\mathbf z\right)W\left(e^{t\bfZ}\bfz \right).
\end{equation}
Note that the complete positivity of the linear maps $\Phi_t$ is equivalent to the quantum diffusion matrix
\begin{equation}\label{eq:CZ}
 \bfC_{\bfZ}:=\bfC-i(\bfZ^{\mathsf{T}}\bfJ+\bfJ\bfZ)\geq0.
\end{equation}
For more details on Gaussian quantum Markov semigroups, we refer to \cite{FG26}.

The generalized Gorini-Kossakowski-Sudarshan-Lindblad (GKLS) representation operators of a Gaussian quantum Markov semigroup are
\begin{align}
 H&=\sum_{j,k=1}^d\left(
O_{jk}\bfa_j^\dagger \bfa_k+\frac12K _{jk}\bfa_j^\dagger \bfa_k^\dagger
 +\frac12\overline{K_{jk}}\bfa_j \bfa_k\right)
 +\frac12\sum_{j=1}^d\left(\zeta_j \bfa_j^\dagger+\overline\zeta_j \bfa_j\right),\label{eq:H}\\
 L_\ell&=\sum_{j=1}^d\left(\overline{v_{\ell j}}\bfa_j+u_{\ell j}\bfa_j^\dagger\right),
 \qquad \ell=1,\ldots,r,\label{eq:Ll}
\end{align}
where $U=(u_{\ell j})_{\ell,j},\ V=(v_{\ell j})_{\ell,j}\in M_{r\times d}(\bC)$, the matrix $O$ is Hermitian and $K$ is symmetric.
The vector $\zeta=(\zeta_1,\ldots,\zeta_d)^{\mathsf T}\in\bC^d$ is related to the real vector in \eqref{eq:covariance} by
\[
 \bfzeta=\begin{pmatrix}\Re\zeta\\ \Im\zeta\end{pmatrix}.
\]
The real-linear operators $Z, C$ on $\mathbb C^d$ corresponding to the bold matrices $\bfZ, \bfC$ are
\begin{align}
 Z z&=\left(\frac12(U^{\mathsf{T}}\overline U-V^{\mathsf{T}}\overline V)+iO\right)z
 +\left(\frac12(U^{\mathsf{T}} V-V^{\mathsf{T}} U)+iK \right)\overline z,\label{eq:ZUV}\\
 Cz&=(U^{\mathsf{T}}\overline U+V^{\mathsf{T}}\overline V)z+(U^{\mathsf{T}}V+V^{\mathsf{T}} U)\overline z.\label{eq:CUV}
\end{align}
Suppose $A z=A_1z+A_2\overline z$.
Then its real matrix $\mathbf A$ is
\begin{align*}
 \mathbf A=\begin{pmatrix}
 \Re A_1+\Re A_2&\Im A_2-\Im A_1\\
 \Im A_1+\Im A_2&\Re A_1-\Re A_2
 \end{pmatrix}.
\end{align*}
With the convention $\Phi_t=e^{-t\cL}$, the GKSL generator is
\begin{equation}\label{eq:GKLS}
 -\cL(X)=i[H,X]+\sum_\ell\left(L_\ell^* X L_\ell-\frac12\{L_\ell^*L_\ell,X\}\right).
\end{equation}
Note that $H$ and $L_\ell$ are unbounded, \eqref{eq:GKLS} is understood as the quadratic form on $\cD\times\cD$ specified in \cite[Equation (22)]{FG26}, or on an appropriate generator core.



Throughout this paper, we shall assume that the number of modes is finite, the mean has been centered, and $\{\Phi_t\}_{t\geq 0}$ has a faithful invariant Gaussian density $\rho$. 
Thus
\begin{equation}\label{eq:lyap}
 \bfZ^{\mathsf{T}}\Sigma+\Sigma\bfZ+\bfC=0,\qquad \Sigma+i\bfJ>0,
\end{equation}
and set
\begin{equation}\label{eq:F}
 \bfF=\Sigma+i\bfJ.
\end{equation}
Combining \eqref{eq:CZ} and \eqref{eq:lyap} gives
\begin{equation}\label{eq:Fidentity}
 \bfC_{\bfZ}=-\bfZ^*\bfF-\bfF\bfZ,\qquad
 \bfF-e^{t\bfZ^{\mathsf{T}}}\bfF e^{t\bfZ}=\int_0^t e^{s\bfZ^{\mathsf{T}}}\bfC_{\bfZ} e^{s\bfZ}\,ds.
\end{equation}
We call $\bfZ$ Hurwitz if every eigenvalue has strictly negative real part. 
In the following, we present some equivalent characterizations of Hurwitz matrices.

Following \cite{FG26}, the condition $(\bfC_{\bfZ}\text{-cond})$ means that there is no nonzero $\bfZ$-invariant complex subspace contained in $\ker\bfC_{\bfZ}$.
In the finite-dimensional case, it is equivalent to
\begin{equation}\label{eq:observability}
 \bigcap_{j=0}^{2d-1}\ker(\bfC_{\bfZ}^{1/2}\bfZ^j)=\{0\},
\end{equation}
or to the full column rank of the vertical matrix with blocks $\bfC_{\bfZ}^{1/2}\bfZ^j$. 


\begin{proposition}\label{prop:control}
The following are equivalent:
\begin{enumerate}[(i)]
\item $\bfZ$ is Hurwitz;
\item $(\bfC_{\bfZ}\text{-cond})$ holds;
\item $\displaystyle \int_0^t e^{s\bfZ^{\mathsf{T}}}\bfC_{\bfZ} e^{s\bfZ}\,ds>0$ for every $t>0$;
\item $\displaystyle \norm{\bfF^{1/2}e^{t\bfZ}\bfF^{-1/2}}<1$ for some $t>0$.
\item The Gaussian quantum Markov semigroup $\{\Phi_t\}_{t\geq 0}$ is irreducible.
\end{enumerate}
\end{proposition}
\begin{proof}
By \cite[Proposition 2.2]{FagLi26} and \cite[Theorem 12]{FG26}, we see that $(ii)$, $(iii)$ and $(v)$ are equivalent.
Equation \eqref{eq:Fidentity} proves (iii)$\Rightarrow$(iv). 
Strict contraction of a matrix exponential implies all its eigenvalues have modulus below one, proving (iv)$\Rightarrow$(i). 
Finally, if (i) holds, integration to infinity in \eqref{eq:Fidentity} yields
\begin{align*}
 \bfF=\int_0^\infty e^{s\bfZ^{\mathsf{T}}}\bfC_{\bfZ} e^{s\bfZ}\,ds>0.
\end{align*}
Any vector in the intersection \eqref{eq:observability} would have zero quadratic form against $\bfF$, so (i)$\Rightarrow$(ii).
This completes the proof of the proposition.
\end{proof}
The equivalence with Hurwitz stability in Proposition \ref{prop:control} uses the faithful invariant-state assumption; 
it is not an assertion about arbitrary Gaussian quantum Markov semigroups.

\subsection{$L_p$-Spaces}

For $1\leq p<\infty$, let $L_p(\rho)$ be the completion of $\cB(\cF_s(\bC^d))$ with respect to $\|\cdot\|_{p, \rho}$, where
\begin{align*}
    \|X\|_{p, \rho}=\|X_p\|_p, \quad X_p=\rho^{\frac{1}{2p}}X\rho^{\frac{1}{2p}}.
\end{align*}
For $p=\infty$ use the operator norm. 
The maps $\Phi_t$ extend consistently as contractions on every $L_p(\rho)$, $1\leq p\leq\infty$. 
For each $1\leq p<\infty$, the operators with finite matrix support
in an eigenbasis of $\rho$ form a common dense subspace of $L_p(\rho)$,
contained in $\rho^{1/2}\cB(\cF_s(\bC^d))\rho^{1/2}$.
The Hilbert space $L_2(\rho)$ has the inner product
\begin{equation}\label{eq:KMSip}
 \langle X,Y\rangle_\rho=\Tr(\rho^{1/2}X^*\rho^{1/2}Y),
\end{equation}
which is not the GNS inner product $\Tr(\rho X^*Y)$. 

Let
\begin{equation}\label{eq:SK}
    \Sigma_{\KMS}=\bfM^{\mathsf{T}} \left(\diag \left( \sqrt{s_1^2-1}, \ldots, \sqrt{s_d^2-1}, \sqrt{s_1^2-1}, \ldots, \sqrt{s_d^2-1}\right)\right) \bfM.
\end{equation}
We have that
\begin{equation}\label{eq:SK-intrinsic}
 \bfz^{\mathsf{T}}\Sigma_{\KMS}\bfw
 =\Tr\bigl(\rho^{1/2}R(\bfz)\rho^{1/2}R(\bfw)\bigr).
\end{equation}
For real $\mathbf z,\mathbf w$,
\begin{equation}\label{eq:kernel}
 \langle W(\mathbf z),W(\mathbf w)\rangle_\rho
 =\exp\left(-\frac12\mathbf z^{\mathsf{T}}\Sigma\mathbf z
             -\frac12\mathbf w^{\mathsf{T}}\Sigma\mathbf w
             +\mathbf z^{\mathsf{T}}\Sigma_{\KMS}\mathbf w\right).
\end{equation}
By taking $\displaystyle \widetilde{W}(\bfz)= e^{\mathbf z^{\mathsf{T}}\Sigma\mathbf z/2}W(\mathbf z)$, we have that 
\begin{equation}\label{eq:expkernel}
 \left\langle \widetilde{W}(\mathbf z),\widetilde{W}(\mathbf w)\right\rangle_\rho
 =e^{\mathbf z^{\mathsf{T}}\Sigma_{\KMS}\mathbf w}.
\end{equation}
These identities follow from the Weyl--KMS pairing in
\cite[Appendix~C, Proposition~33]{FPSU25}, with the numbering of the arXiv v1 version.
The calculation is made in thermal coordinates and transferred by the Williamson transformation.
Replacing $z$ by $-z$ in that pairing gives \eqref{eq:kernel}; differentiation at the origin and rescaling give \eqref{eq:SK-intrinsic} and \eqref{eq:expkernel}, respectively.

A normal $\rho$-preserving unital completely positive map $\Phi$ has a normal unital completely positive KMS adjoint $\Phi^\sharp$, characterized by
\begin{equation}\label{eq:adjoint}
 \rho^{1/2}\Phi^\sharp(Y)\rho^{1/2}
 =\Phi^*(\rho^{1/2}Y\rho^{1/2}).
\end{equation}

We say that $\{\Phi_t\}_{t\geq 0}$ is KMS-symmetric if $\Phi_t^\sharp=\Phi_t$.
For $1<p,q<\infty$ and their conjugate exponents $p',q'$, weighted duality gives
\begin{equation}\label{eq:forward-dual}
 \norm{\Phi}_{p\to q,\rho}
 =\norm{\Phi^\sharp}_{q'\to p',\rho}.
\end{equation}

For $0<p<1$, the $p$-functional $\|\cdot\|_{p, \rho}$ is only defined on positive operators, which is not a norm. 
For $a>0$ and a bounded uniformly positive $Y$, the negative-exponent notation is defined unambiguously by
\begin{equation}\label{eq:negative}
 \| Y\|_{-a, \rho}:=\|Y^{-1}\|_{a, \rho}^{-1}.
\end{equation}
No bounded inverse of $\rho$ is assumed anywhere.

Let
\begin{align*}
\bfZ_{\KMS}=\Sigma_{\KMS}^{1/2}\bfZ\Sigma_{\KMS}^{-1/2}.
\end{align*}
For a contraction $T$ on $\bC^{2d}$, write
\[
 \cF_c(T)=\bigoplus_{n=0}^\infty T^{\otimes_s n}.
\]
It acts on exponential vectors by $\cF_c(T)e_{\bfz}=e_{T\bfz}$.

\begin{proposition}\label{prop:second}
There is a unitary $\mathcal U:L_2(\rho)\to\cF_s(\mathbb C^{2d})$ with
\begin{equation}\label{eq:U}
 \mathcal U  \widetilde{W}(\mathbf z)=e_{\Sigma_{\KMS}^{1/2}\mathbf z},
\end{equation}
and
\begin{equation}\label{eq:second}
 \mathcal U \Phi_t^{(2)}\mathcal U^{-1}=\cF_c(e^{t\bfZ_{\KMS}}),
\end{equation}
where $\Phi_t^{(2)}$ is viewed as bounded operator on $L^2(\rho)$ extending $\Phi_t$.
\end{proposition}
\begin{proof}
By Equation \eqref{eq:expkernel}, we see that $\mathcal U$ is an isometry on the span of normalized Weyl vectors. 
That span is dense in $L_2(\rho)$. 
On the auxiliary Fock space, exponential vectors with parameters in the real subspace $\mathbb R^{2d}$ are total:
differentiating an entire exponential-vector matrix coefficient on this real subspace determines every symmetric tensor coefficient. Hence the isometry is onto.

Stationarity of the semigroup gives $\displaystyle \bfC_t=\Sigma-e^{t\bfZ^\mathsf{T}}\Sigma e^{t\bfZ}$. Consequently
\begin{align}\label{eq: exchange relation for Gaussian QMS}
 \Phi_t \left(\widetilde{W}(\mathbf z)\right)=\widetilde{W}(e^{t\bfZ}\mathbf z).
\end{align}
Applying \eqref{eq:U} proves \eqref{eq:second} on a total set.
Since $\Phi_t$ is contractive on $L_2(\rho)$, restriction to the first chaos shows $\|e^{t\bfZ_{\KMS}}\|\leq1$.
\end{proof}

\begin{remark}
 Taking Hilbert-space adjoints, we have that 
    \begin{equation}\label{eq:secondadj}
 \mathcal U(\Phi_t^\sharp)^{(2)}\mathcal U^{-1}
 =\cF_c(e^{t\bfZ_{\KMS}^{\mathsf{T}}}).
\end{equation}
\end{remark}

Let
\begin{equation}\label{eq:CK}
 \bfC_{\KMS}=-\bfZ^{\mathsf{T}}\Sigma_{\KMS}-\Sigma_{\KMS}\bfZ.
\end{equation}
\begin{proposition}\label{prop:CK}
We have that $\bfC_{\KMS}\geq 0$.
Moreover, if
\begin{equation}\label{eq:gammaK}
\gamma_{\KMS}
=\lambda_{\min}\left(\Sigma_{\KMS}^{-1/2}\bfC_{\KMS}\Sigma_{\KMS}^{-1/2}\right)>0,
\end{equation}
then
\begin{equation}\label{eq:oneparticle-bound}
 \|e^{t\bfZ_{\KMS}}\|, \|e^{t\bfZ_{\KMS}^*}\|\leq e^{-\gamma_{\KMS} t/2}.
\end{equation}
\end{proposition}
\begin{proof}
Differentiate $\|e^{t\bfZ_{\KMS}}\bfz\|^2\leq\| \bfz\|^2$ at zero, we see that $\bfZ_{\KMS}+\bfZ_{\KMS}^{\mathsf{T}}\leq0$, equivalently $\bfC_{\KMS} \geq 0$. 
Under \eqref{eq:gammaK}, $\bfZ_{\KMS}+\bfZ_{\KMS}^{\mathsf{T}}\leq-\gamma_{\KMS} I$, and differentiation followed by Gronwall inequality proves \eqref{eq:oneparticle-bound}.
\end{proof}




\subsection{GNS and KMS symmetry}
In this section, we shall investigate GNS and KMS symmetry for Gaussian quantum Markov semigroups.
The KMS spectral gap of Gaussian quantum Markov semigroups has been investigated by Li~\cite{Li26}.
For related algebraic characterizations of GNS symmetry, see Poletti, Sasso, and Umanit\`a~\cite{PSU26}.

\begin{proposition}\label{prop:kmssym}
Suppose that $\{\Phi_t\}_{t\geq 0}$ is a Gaussian quantum Markov semigroup.
Then the following are equivalent:
\begin{enumerate}
    \item $\{\Phi_t\}_{t \geq 0}$ is KMS-symmetric;
    \item $\bfZ^{\mathsf{T}}\Sigma_{\KMS}=\Sigma_{\KMS}\bfZ$;
    \item $\bfZ_{\KMS}$ is Hermitian.
\end{enumerate}
\end{proposition}

\begin{proof}
Self-adjointness of the full semigroup implies self-adjointness of its first-chaos restriction. Conversely, if $\bfZ_{\KMS}=\bfZ_{\KMS}^*$, every tensor power of $e^{t\bfZ_{\KMS}}$ is self-adjoint, so Proposition~\ref{prop:second} gives KMS symmetry.
\end{proof}

\begin{proposition}\label{prop:hurwitzpos}
Suppose that $\{\Phi_t\}_{t\geq 0}$ is a KMS-symmetric Gaussian quantum Markov semigroup and $\bfZ$ is Hurwitz.
Then
\begin{equation}\label{eq:symmetric-rate}
 \gamma_{\KMS} =-2\lambda_{\max}(\bfZ_{\KMS})>0.
\end{equation}   
\end{proposition}
\begin{proof}
    Similarity preserves eigenvalues; thus a symmetric $\bfZ_{\KMS}$ similar to a Hurwitz matrix is strictly negative definite.
\end{proof}

\begin{proposition}
Suppose that $\{\Phi_t\}_{t\geq 0}$ is a Gaussian quantum Markov semigroup.
Then $\{\Phi_t\}_{t\geq 0}$ is GNS-symmetric if and only if $\mathbf{Z}^\mathsf{T}\mathbf{F}=\mathbf{F}\mathbf{Z}$.
\end{proposition}

\begin{proof}
    Using $W(z)^*=W(-z)$ and the Weyl multiplication equation \eqref{eq:Weyl multiplication convention}, we obtain
\begin{align}\label{eq: GNS isometry}
\big\langle\widetilde W(z),
\widetilde W(w)\big\rangle_{\mathrm{GNS},\rho}
&=\exp\!\bigl(
\mathbf z^{\mathsf{T}}\Sigma\mathbf w
+i\mathbf z^{\mathsf{T}}\mathbf J\mathbf w
\bigr)
=\exp\!\bigl(\mathbf z^{\mathsf{T}}\mathbf F\mathbf w\bigr).
\end{align}
Combining with \eqref{eq: exchange relation for Gaussian QMS}, we have
\begin{align*}
\big\langle\widetilde W(z),
\Phi_t(\widetilde W(w))\big\rangle_{\mathrm{GNS},\rho}
&=\exp\!\bigl(
\mathbf z^{\mathsf{T}}\mathbf F e^{t\mathbf Z}\mathbf w
\bigr),
\\
\big\langle\Phi_t(\widetilde W(z)),
\widetilde W(w)\big\rangle_{\mathrm{GNS},\rho}
&=\exp\!\bigl(
\mathbf z^{\mathsf{T}}e^{t\mathbf Z^{\mathsf{T}}}\mathbf F\mathbf w
\bigr).
\end{align*}
The above inner product is the GNS inner product induced by $\rho$. Since the linear span of the elements of the form $\widetilde W(z)$ is dense,
the proof of the proposition is complete.
\end{proof}

\section{Complete Modified Logarithmic Sobolev Inequalities}
\label{sec:complete-mlsi}


In this section, we determine the optimal relative complete
modified logarithmic Sobolev constant for finite-mode Gaussian
quantum Markov semigroups.
We give an explicit matrix formula for this constant, identify
it with twice the optimal relative GNS norm-decay rate, and
characterize its strict positivity.
We also determine the optimal contraction coefficient at each fixed time and show that the positivity criterion holds without assuming convergence of the drift evolution.

Whenever it exists, let $\mathbb E_\Phi$ denote the normal
$\rho$-preserving conditional expectation onto the fixed-point algebra
of $\{\Phi_t\}_{t \ge 0}$.
We say that $\Phi$ satisfies CMLSI with rate $\beta \geq 0$ if
\begin{equation}\label{eq:section4-cmlsi}
H\!\left(
(\Phi_t^*\otimes\Id_k)(\omega)
\,\middle\|\,
(\mathbb E_{\Phi}^*\otimes\Id_k)(\omega)
\right)
\leq
e^{-\beta t}
H\!\left(
\omega
\,\middle\|\,
(\mathbb E_{\Phi}^*\otimes\Id_k)(\omega)
\right)
\end{equation}
for every $t\geq0$, every $k\geq1$, and every density $\omega$
for which the relative entropy on the right is finite.
Here $H$ denotes Umegaki relative entropy.
The supremum of all admissible rates is denoted by
$\alpha_{\mathrm{CMLSI}}(\Phi;\mathbb E_\Phi)$, or simply
$\alpha_{\mathrm{CMLSI}}$.
We normalize the CMLSI constant as the complete relative entropy decay rate.
Thus the convention with decay $e^{-2\alpha t}$ assigns half this value.
The rate $\beta=0$ is always admissible by data processing.
We say that $\Phi$ satisfies a positive relative CMLSI if
$\alpha_{\mathrm{CMLSI}}(\Phi;\mathbb E_\Phi)>0$.

We also write
$\lambda_{\mathrm{GNS}}(\Phi;\mathbb E_\Phi)$
for the largest $\lambda\geq0$ such that
\begin{equation}\label{eq:lambda GNS}
\|\Phi_t(X)\|_{\mathrm{GNS},\rho}
\leq e^{-\lambda t}\|X\|_{\mathrm{GNS},\rho},
\qquad
t\geq0,\quad \mathbb E_\Phi(X)=0.
\end{equation}

\begin{lemma} \label{lem:section4-fixed-splitting}
Suppose that
$\displaystyle \mathbf P=\lim_{t\to\infty}e^{t\mathbf Z}$
exists. There is a symplectic change of Weyl coordinates such that,
after grouping the coordinates by sector,
\begin{equation}\label{eq:section4-fixed-blocks}
\begin{aligned}
\mathbf Z&=0_{2l}\oplus\mathbf Z_{\mathrm r},
&
\Sigma&=\Sigma_0\oplus\Sigma_{\mathrm r},\\
\mathbf F&=\mathbf F_0\oplus\mathbf F_{\mathrm r},
&
\bfC_{\bfZ}&=0_{2l}\oplus\bfC_{\bfZ,\mathrm r},
\end{aligned}
\end{equation}
where $\mathbf Z_{\mathrm r}$ is Hurwitz whenever the relaxing sector
is nonempty, and every nonempty covariance block is faithful.
Moreover, $\mathbf C=0_{2l}\oplus\mathbf C_{\mathrm r}$.

In the corresponding Gaussian-unitary realization,
\[
\Phi_t
=
\Id_{\mathcal B(\mathcal F_s(\mathbb C^l))}
\otimes\Phi_{t,\mathrm r},
\qquad
\rho=\rho_0\otimes\rho_{\mathrm r},
\]
where
$\rho_0:=\operatorname{Tr}_{\mathsf h_{\mathrm r}}(\rho),
\;
\rho_{\mathrm r}:=\operatorname{Tr}_{\mathsf h_0}(\rho)$
in the representation
$\mathsf h=\mathsf h_0\otimes\mathsf h_{\mathrm r}$.
Its fixed-point conditional expectation is
\[
\mathbb E_\Phi(X_0\otimes X_{\mathrm r})
=
\operatorname{Tr}(\rho_{\mathrm r}X_{\mathrm r})
X_0\otimes1_{\mathrm r},
\qquad
\mathbb E_{\Phi}^*(\omega)
=
\operatorname{Tr}_{\mathrm r}(\omega)\otimes\rho_{\mathrm r}.
\]
Empty sectors are omitted.
\end{lemma}

\begin{proof}
Let $\mathcal N=\ker_{\mathbb R}\mathbf Z$.
For $x\in\mathcal N$,
$x^{\mathsf{T}}\bfC_{\bfZ}x=0$ and $\bfC_{\bfZ}\geq0$
implies that
$\bfC_{\bfZ}x=0,
\;
\mathbf Z^*\mathbf F x=0$.
Taking real and imaginary parts yields
$\mathbf Z^*\Sigma x=0,
\;
\mathbf Z^*\mathbf Jx=0$.
By invertibility of $\Sigma$ and $\mathbf J$, and by
$\dim\ker_{\mathbb R}\mathbf Z
=\dim\ker_{\mathbb R}\mathbf Z^*$, we obtain
$\Sigma\mathcal N
=
\mathbf J\mathcal N
=
\ker_{\mathbb R}\mathbf Z^*$.

The symplectic form is nondegenerate on $\mathcal N$.
Indeed, suppose that $x\in\mathcal N$ satisfies
$x^{\mathsf{T}}\mathbf Jy=0$ for every $y\in\mathcal N$.
Choose $y\in\mathcal N$ such that $\mathbf Jy=\Sigma x$.
Then
$0=x^{\mathsf{T}}\mathbf Jy=x^{\mathsf{T}}\Sigma x$
shows $x=0$.
Hence $\dim\mathcal N=2l$ for some $l$, and
$\mathbb R^{2d}
=
\mathcal N\oplus\mathcal N^{\perp_{\mathbf J}}$.

Both summands are invariant under $\mathbf Z$.
Indeed, $\mathbf Z\mathcal N=0$, while
$\mathbf Z^*\mathbf Jx=0$ for $x\in\mathcal N$ implies
$x^{\mathsf{T}}\mathbf J\mathbf Zy=0$
for every $y$. Thus
$\operatorname{Ran}\mathbf Z\subseteq
\mathcal N^{\perp_{\mathbf J}}$.
The restriction of $\mathbf Z$ to
$\mathcal N^{\perp_{\mathbf J}}$ is invertible because its kernel is
contained in
$\mathcal N\cap\mathcal N^{\perp_{\mathbf J}}=\{0\}$.
Existence of $\displaystyle \lim_{t\to\infty}e^{t\mathbf Z}$ therefore implies that
this restriction is Hurwitz.

The covariance has no cross terms between the two summands.
For $x\in\mathcal N$, write $\Sigma x=\mathbf Jy$ with
$y\in\mathcal N$. If $v\in\mathcal N^{\perp_{\mathbf J}}$, then
$x^{\mathsf{T}}\Sigma v
=
(\mathbf Jy)^{\mathsf{T}}v
=
-y^{\mathsf{T}}\mathbf Jv=0$.

Choose symplectic bases of $\mathcal N$ and
$\mathcal N^{\perp_{\mathbf J}}$, and let
$S=(S_0\ S_{\mathrm r})$
be the resulting change-of-basis matrix. Then
\[
S^{\mathsf{T}}\mathbf J S
=
\mathbf J_0\oplus\mathbf J_{\mathrm r},
\quad
\mathbf Z':=S^{-1}\mathbf ZS
=
\begin{pmatrix}
0&0\\
0&\mathbf Z_{\mathrm r}
\end{pmatrix}.
\]
The previously established identity
$x^{\mathsf{T}}\Sigma v=0$ for
$x\in\mathcal N$ and $v\in\mathcal N^{\perp_{\mathbf J}}$
implies $S_0^{\mathsf{T}}\Sigma S_{\mathrm r}=0$.
Since $\Sigma$ is symmetric,
\[
\Sigma':=S^{\mathsf{T}}\Sigma S
=
\begin{pmatrix}
\Sigma_0&0\\
0&\Sigma_{\mathrm r}
\end{pmatrix},
\qquad
\mathbf F':=S^{\mathsf{T}}\mathbf F S
=
\begin{pmatrix}
\mathbf F_0&0\\
0&\mathbf F_{\mathrm r}
\end{pmatrix},
\]
where
$\mathbf F_0=\Sigma_0+i\mathbf J_0,
\qquad
\mathbf F_{\mathrm r}
=\Sigma_{\mathrm r}+i\mathbf J_{\mathrm r}$.

Setting $\mathbf{C}'=S^{\mathsf{T}}\mathbf{C}S$ and
$\bfC_{\bfZ}'=S^{\mathsf{T}}\bfC_{\bfZ}S$, the stationary
identities become
\[
\mathbf{C}'=-\mathbf Z'^*\Sigma'-\Sigma'\mathbf Z',
\qquad
\bfC_{\bfZ}'=-\mathbf Z'^*\mathbf F'-\mathbf F'\mathbf Z'.
\]
Substituting the preceding block matrices gives
\[
\mathbf{C}'=0_{2l}\oplus\mathbf{C}_{\mathrm r},
\qquad
\bfC_{\bfZ}'=0_{2l}\oplus\bfC_{\bfZ,\mathrm r},
\]
with
$\mathbf{C}_{\mathrm r}
=-\mathbf Z_{\mathrm r}^*\Sigma_{\mathrm r}
-\Sigma_{\mathrm r}\mathbf Z_{\mathrm r},
\;
\bfC_{\bfZ,\mathrm r}
=-\mathbf Z_{\mathrm r}^*\mathbf F_{\mathrm r}
-\mathbf F_{\mathrm r}\mathbf Z_{\mathrm r}$.

In the corresponding Gaussian-unitary realization,
$W(\mathbf z_0\oplus\mathbf z_{\mathrm r})
=W_0(\mathbf z_0)\otimes W_{\mathrm r}(\mathbf z_{\mathrm r})$.
The block forms of $\mathbf Z$ and $\mathbf C$ in
\eqref{eq:Tweyl} give
$\Phi_t=\Id\otimes\Phi_{t,\mathrm r}$,
while the block form of $\Sigma$ in \eqref{eq:char} gives
$\rho=\rho_0\otimes\rho_{\mathrm r}$.
The factors are faithful because $\mathbf F_0$ and
$\mathbf F_{\mathrm r}$ are positive definite whenever nonempty.
Since $\mathbf Z_{\mathrm r}$ is Hurwitz,
\cite[Theorem~17]{GP26} identifies the fixed-point algebra as
$\mathcal B(\mathsf h_0)\otimes1_{\mathrm r}$ and the asymptotic
conditional expectation as
$\Id\otimes\operatorname{Tr}(\rho_{\mathrm r}\,\cdot)1_{\mathrm r}$.
This gives the stated formulas for $\mathbb E_\Phi$ and its predual.
\end{proof}

\begin{theorem}\label{thm:section4-relative-cmlsi}
Under the assumptions of Lemma~\ref{lem:section4-fixed-splitting},
suppose that $l<d$, so that the relaxing sector contains at
least one mode. Then
\begin{equation}\label{eq:section4-cmlsi-bound}
\begin{aligned}
\alpha_{\mathrm{CMLSI}}(\Phi;\mathbb E_\Phi)
=
\lambda_{\min}\!\left(
\mathbf F_{\mathrm r}^{-1/2}
\bfC_{\bfZ,\mathrm r}
\mathbf F_{\mathrm r}^{-1/2}
\right)
=
2\lambda_{\mathrm{GNS}}(\Phi;\mathbb E_\Phi).
\end{aligned}
\end{equation}
The optimal rate admits the equivalent matrix characterizations
\begin{equation}\label{eq:section4-cmlsi-order}
\begin{split}
\lambda_{\min}\!\left(
\mathbf F_{\mathrm r}^{-1/2}
\bfC_{\bfZ,\mathrm r}
\mathbf F_{\mathrm r}^{-1/2}
\right)
=
\sup\left\{
\beta\geq0:
\bfC_{\bfZ}\geq
\beta(\mathbf F-\mathbf P^*\mathbf F\mathbf P)
\right\},
\end{split}
\end{equation}
and
\begin{equation}\label{eq:section4-cmlsi-rayleigh}
\begin{split}
\lambda_{\min}\!\left(
\mathbf F_{\mathrm r}^{-1/2}
\bfC_{\bfZ,\mathrm r}
\mathbf F_{\mathrm r}^{-1/2}
\right)
=
\inf_{\mathbf z\notin\ker_{\mathbb C}\mathbf Z}
\frac{
\mathbf z^*\bfC_{\bfZ}\mathbf z
}{
\operatorname{dist}_{\mathbf F}
(\mathbf z,\ker_{\mathbb C}\mathbf Z)^2
},
\end{split}
\end{equation}
where $\operatorname{dist}_{\mathbf F}$ denotes the distance
induced by the inner product
$\langle u,v\rangle_{\mathbf F}=u^*\mathbf Fv$.
In particular,
\begin{equation}\label{eq:section4-positive-matrix-rate}
\lambda_{\min}\!\left(
\mathbf F_{\mathrm r}^{-1/2}
\bfC_{\bfZ,\mathrm r}
\mathbf F_{\mathrm r}^{-1/2}
\right)>0
\quad\text{ if and only if }\quad
\ker_{\mathbb C}\bfC_{\bfZ}
=
\ker_{\mathbb C}\mathbf Z.
\end{equation}
Whenever these equivalent conditions hold, $\Phi$ satisfies a
positive CMLSI relative to $\mathbb E_\Phi$.
No GNS or KMS symmetry is required.
\end{theorem}

The lower bound follows from Gaussian comparison with a reference
semigroup in Lemmas~\ref{lem:section4-reference-entropy}
and~\ref{lem:section4-gaussian-comparison}.
The upper bound follows from the entropy-to-GNS implication in
Lemma~\ref{lem:section4-entropy-implies-gns} and the exact GNS rate
in Lemma~\ref{lem:section4-relative-gns-rate}, which also gives
the matrix characterizations and positivity criterion.

In what follows, $\mathsf h$ denotes the system Hilbert space
and $\mathsf k$ an auxiliary Hilbert space.
For a joint density operator $\omega$ on
$\mathsf h\otimes\mathsf k$, we write
\[
\omega_{\mathsf h}:=\operatorname{Tr}_{\mathsf k}\omega,
\qquad
\omega_{\mathsf k}:=\operatorname{Tr}_{\mathsf h}\omega
\]
for the corresponding reduced density operators.

The corresponding single-mode thermal estimate was proved by
De Palma and Huber in \cite[Theorem~9]{DH18}.
We obtain the finite-mode Gaussian version needed here from the
conditional entropy power inequality of De Palma and Trevisan
\cite[Theorem~6]{DP18}, together with a
finite-dimensional approximation of the auxiliary system.
\begin{lemma}
\label{lem:section4-reference-entropy}
Let $\rho$ be a faithful centered Gaussian density on a finite-mode
bosonic system, with covariance $\Sigma$. 
Define
\[
\Psi_s(W(z))
=
\exp\!\left[
-\frac12(1-e^{-s})\mathbf z^{\mathsf T}\Sigma\mathbf z
\right]W(e^{-s/2}z),
\qquad s\geq0.
\]
For every separable auxiliary Hilbert space $\mathsf k$ and every
joint density $\omega$ with
$H(\omega\|\rho\otimes\omega_{\mathsf k})<\infty$, we have that 
\[
H\!\left(
(\Psi_s^*\otimes\Id_{\mathsf k})(\omega)
\,\middle\|\,
\rho\otimes\omega_{\mathsf k}
\right)\leq
e^{-s}H(\omega\|\rho\otimes\omega_{\mathsf k}).
\]
\end{lemma}

\begin{proof}
The case $s=0$ is immediate.
First suppose that $\mathsf k$ is finite-dimensional.
By the Gaussian normal form \eqref{eq:Williamson decomposition}, $\Tr(\sqrt{\rho})<\infty$ and 
\begin{equation}\label{eq:- log rho expansion}
-\log\rho
=
cI+\frac12\mathbf R^{\mathsf T}G\mathbf R,
\qquad
\mathbf R=(Q_1,\ldots,Q_d,P_1,\ldots,P_d)^{\mathsf T},
\end{equation}
where $c\in\mathbb R$ and $G=G^{\mathsf T}>0$ is a real
$2d\times2d$ matrix. 
Applying the Gibbs variational inequality
\cite[Proposition~1]{PD88} to bounded spectral truncations of
$-\log\rho$, followed by monotone convergence and data processing,
gives
\[
\begin{split}
\frac12\operatorname{Tr}
\bigl(\omega_{\mathsf h}(-\log\rho)\bigr)
&\leq
H(\omega_{\mathsf h}\|\rho)
+\log\operatorname{Tr}\sqrt{\rho}\\
&\leq
H(\omega\|\rho\otimes\omega_{\mathsf k})
+\log\operatorname{Tr}\sqrt{\rho}
<\infty.
\end{split}
\]
Let $\displaystyle \mathbf{N}=\sum_{j=1}^d \mathbf {a_j^\dagger a_j}$ be the total number operator.
By \eqref{eq:- log rho expansion}, we have $-\log\rho\geq a\mathbf N-bI$ for some $a>0$ and $b\geq0$, then
$a\operatorname{Tr}(\omega_{\mathsf h} \mathbf{N})
\leq
\operatorname{Tr}\bigl(\omega_{\mathsf h}(-\log\rho)\bigr)+b
<\infty$.
The Gibbs variational principle yields the following von Neumann entropy bound:
\[
S(\omega_{\mathsf h})
\leq
\operatorname{Tr}(\omega_{\mathsf h}\mathbf N)
+\log\operatorname{Tr}(e^{- \mathbf N})
=
\operatorname{Tr}(\omega_{\mathsf h}\mathbf N)
-d\log(1-e^{-1})
<\infty.
\]
By the subadditivity and $\dim\mathsf k<\infty$,
\[
S(\omega)
\leq
S(\omega_{\mathsf h})+S(\omega_{\mathsf k})
\leq
S(\omega_{\mathsf h})+\log\dim\mathsf k
<\infty.
\]

Let $\mathsf h'$ be a copy of $\mathsf h$.
The environment $\mathsf h'$ is prepared in the state $\rho$,
independently of the input.
Let $U_s$ be the beam-splitter unitary on
$\mathsf h\otimes\mathsf h'$ with transmissivity $e^{-s}$.
Then
\[
\Psi_s^*(\omega_{\mathsf h})
=
\operatorname{Tr}_{\mathsf h'}
\left[
U_s(\omega_{\mathsf h}\otimes\rho)U_s^*
\right],
\]
and
\[
U_s^*(\mathbf R\otimes I_{\mathsf h'})U_s
=
e^{-s/2}(\mathbf R\otimes I_{\mathsf h'})
+\sqrt{1-e^{-s}}(I_{\mathsf h}\otimes\mathbf R).
\]
Using \eqref{eq:- log rho expansion} and the fact $\rho$ is centered,
we obtain
\begin{equation}\label{eq: energy eq}
\begin{aligned}
&\operatorname{Tr}
\bigl((-\log\rho)\Psi_s^*(\omega_{\mathsf h})\bigr)\\
&\quad=
\operatorname{Tr}\left[
(\omega_{\mathsf h}\otimes\rho)
U_s^*\bigl((-\log\rho)\otimes I_{\mathsf h'}\bigr)U_s
\right]\\
&\quad=
c+\frac12\sum_{j,k=1}^{2d}G_{jk}
\operatorname{Tr}\left[
(\omega_{\mathsf h}\otimes\rho)
U_s^*(R_jR_k\otimes I_{\mathsf h'})U_s
\right]\\
&\quad=
e^{-s}\left[
c+\frac12\operatorname{Tr}
\bigl(\omega_{\mathsf h}\mathbf R^{\mathsf T}G\mathbf R\bigr)
\right]
+(1-e^{-s})\left[
c+\frac12\operatorname{Tr}
\bigl(\rho\mathbf R^{\mathsf T}G\mathbf R\bigr)
\right]\\
&\quad=
e^{-s}\operatorname{Tr}
\bigl((-\log\rho)\omega_{\mathsf h}\bigr)
+(1-e^{-s})S(\rho)<\infty.
\end{aligned}
\end{equation}

The Gibbs entropy bound then yields
\[
S\bigl(\Psi_s^*(\omega_{\mathsf h})\bigr)
\leq
\operatorname{Tr}\bigl(\mathbf N\Psi_s^*(\omega_{\mathsf h})\bigr)
+\log\operatorname{Tr}(e^{- \mathbf N})
<\infty.
\]
Together with $\dim\mathsf k<\infty$, this ensures that the
input and output conditional entropies are finite.

Let $I(\mathsf h:\mathsf h'\mid\mathsf k)$ denote the quantum
conditional mutual information between $\mathsf h$ and
$\mathsf h'$ given $\mathsf k$, and let
$S(\mathsf h'\mid\mathsf k)$ denote the conditional von Neumann
entropy of $\mathsf h'$ given $\mathsf k$.
Subscripts specify the state in which these quantities are
evaluated.
The remaining notation is defined analogously; see
\cite[Sections~11.4 and~11.7]{WildeMark17}.
\[
\begin{aligned}
I(\mathsf h:\mathsf h'\mid\mathsf k)_{\omega\otimes\rho}=
S(\omega)
+S(\rho\otimes\omega_{\mathsf k})
-S(\omega_{\mathsf k})
-S(\omega\otimes\rho)=0.
\end{aligned}
\]
Moreover,
\[
\begin{aligned}
S(\mathsf h'\mid\mathsf k)_{\rho\otimes\omega_{\mathsf k}}
=
S(\rho\otimes\omega_{\mathsf k})-S(\omega_{\mathsf k})=S(\rho).
\end{aligned}
\]
For $s>0$, the linear conditional entropy power inequality
\cite[Theorem~6, Eq.~(106)]{DP18}, with its parameters
$\lambda,\eta$ both chosen equal to the transmissivity $e^{-s}$,
therefore gives
\[
S(\mathsf h\mid\mathsf k)_{
(\Psi_s^*\otimes\Id_{\mathsf k})(\omega)}
\geq
e^{-s}S(\mathsf h\mid\mathsf k)_\omega
+(1-e^{-s})S(\rho).
\]
Here
\[
\begin{aligned}
S(\mathsf h\mid\mathsf k)_\omega
&=S(\omega)-S(\omega_{\mathsf k}),\\
S(\mathsf h\mid\mathsf k)_{
(\Psi_s^*\otimes\Id_{\mathsf k})(\omega)}
&=
S\bigl((\Psi_s^*\otimes\Id_{\mathsf k})(\omega)\bigr)
-S(\omega_{\mathsf k}),
\end{aligned}
\]
because the channel leaves the marginal on $\mathsf k$ unchanged.

Combining with \eqref{eq: energy eq}, we have
\[
\begin{aligned}
&H\!\left(
(\Psi_s^*\otimes\Id_{\mathsf k})(\omega)
\,\middle\|\,
\rho\otimes\omega_{\mathsf k}
\right)\\
&\quad=
\operatorname{Tr}
\bigl((-\log\rho)\Psi_s^*(\omega_{\mathsf h})\bigr)
-
S(\mathsf h\mid\mathsf k)_{
(\Psi_s^*\otimes\Id_{\mathsf k})(\omega)}\\
&\quad\leq
e^{-s}\operatorname{Tr}
\bigl((-\log\rho)\omega_{\mathsf h}\bigr)
+(1-e^{-s})S(\rho)\\
&\qquad\quad
-e^{-s}S(\mathsf h\mid\mathsf k)_\omega
-(1-e^{-s})S(\rho)\\
&\quad=
e^{-s}\left[
\operatorname{Tr}
\bigl((-\log\rho)\omega_{\mathsf h}\bigr)
-S(\mathsf h\mid\mathsf k)_\omega
\right]\\
&\quad=
e^{-s}H(\omega\|\rho\otimes\omega_{\mathsf k}).
\end{aligned}
\]
This proves the estimate when $\mathsf k$ is finite-dimensional.
If $\omega_{\mathsf k}$ is not faithful, the logarithmic
identities are understood after restricting the auxiliary
space to the support of $\omega_{\mathsf k}$.

For a general separable $\mathsf k$, choose finite-rank projections
$P_n\uparrow I_{\mathsf k}$ and a unit vector
$e_0\in\operatorname{Ran}P_1$. The channels
\[
\Theta_n(\tau)
=
P_n\tau P_n
+
\operatorname{Tr}((I_{\mathsf k}-P_n)\tau)
|e_0\rangle\langle e_0|
\]
have finite-dimensional output and satisfy
\[
(\Id_{\mathsf h}\otimes\Theta_n)(\omega)
\longrightarrow\omega
\qquad\text{in trace norm}.
\]
Apply the finite-dimensional estimate to these joint densities.
By the data processing inequality, their input relative entropies with respect to $\rho\otimes\Theta_n(\omega_{\mathsf k})$ are bounded above by
$H(\omega\|\rho\otimes\omega_{\mathsf k})$.
The output densities and their corresponding reference densities
converge in trace norm, by trace-norm contractivity.
Joint lower semicontinuity of relative entropy therefore yields
the asserted estimate.
\end{proof}

\begin{lemma}
\label{lem:section4-gaussian-comparison}
Let $\{\Phi_t\}_{t\geq 0}$ be a finite-mode Gaussian QMS preserving a faithful
centered Gaussian density $\rho$, with coefficients
$\mathbf Z,\mathbf C$ and covariance $\Sigma$.
If $\bfC_{\bfZ}\geq\beta\mathbf F$ for some $\beta\geq0$,
then there is a normal $\rho$-preserving Gaussian quantum Markov semigroup
$\{\mathcal R_t\}_{t\geq 0}$ such that
\[
\Phi_t
=\mathcal R_t\circ\Psi_{\beta t}
=\Psi_{\beta t}\circ\mathcal R_t,
\]
where $\Psi$ is the reference semigroup of
Lemma~\ref{lem:section4-reference-entropy}. Consequently, for every
separable auxiliary Hilbert space $\mathsf k$,
\[
H\!\left(
(\Phi_t^*\otimes\Id_{\mathsf k})(\omega)
\,\middle\|\,
\rho\otimes\omega_{\mathsf k}
\right)\leq
e^{-\beta t}H(\omega\|\rho\otimes\omega_{\mathsf k}),
\]
whenever the relative entropy on the right is finite.
\end{lemma}

\begin{proof}
For $\beta=0$, the asserted entropy estimate follows from data
processing, and one may take $\mathcal R_t=\Phi_t$.
Assume henceforth that $\beta>0$.
The Gaussian coefficients
\[
\mathbf Z+\frac{\beta}{2}I_{2d},
\qquad
\mathbf C-\beta\Sigma
\]
satisfy
\[
\begin{split}
&(\mathbf C-\beta\Sigma)
-i\left[
\left(\mathbf Z+\frac{\beta}{2}I_{2d}\right)^*\mathbf J
+\mathbf J\left(\mathbf Z+\frac{\beta}{2}I_{2d}\right)
\right]\\
&\hspace{35mm}=\bfC_{\bfZ}-\beta\mathbf F\geq0,
\end{split}
\]
and
\[
\left(\mathbf Z+\frac{\beta}{2}I_{2d}\right)^*\Sigma
+\Sigma\left(\mathbf Z+\frac{\beta}{2}I_{2d}\right)
+(\mathbf C-\beta\Sigma)=0.
\]
The Gaussian complete-positivity criterion \eqref{eq:covariance},\eqref{eq:CZ} from Section~2 therefore
gives a normal $\rho$-preserving unital completely positive
Gaussian semigroup $\mathcal R_t$ with these coefficients.
Its Weyl action is
\[
\mathcal R_t(W(\mathbf z))
=
\exp\!\left[
-\frac12\mathbf z^{\mathsf T}
\left(
\Sigma-e^{\beta t}e^{t\mathbf Z^*}\Sigma e^{t\mathbf Z}
\right)\mathbf z
\right]\times
W(e^{\beta t/2}e^{t\mathbf Z}\mathbf z).
\]
Direct composition on Weyl operators yields
\[
\Phi_t
=
\mathcal R_t\circ\Psi_{\beta t}
=
\Psi_{\beta t}\circ\mathcal R_t.
\]
By Lemma~\ref{lem:section4-reference-entropy} and the data processing inequality, we have that 
\[
\begin{split}
&H\!\left(
(\Phi_t^*\otimes\Id_{\mathsf k})(\omega)
\,\middle\|\,
\rho\otimes\omega_{\mathsf k}
\right)\\
&\quad\leq
e^{-\beta t}
H\!\left(
(\mathcal R_t^*\otimes\Id_{\mathsf k})(\omega)
\,\middle\|\,
\rho\otimes\omega_{\mathsf k}
\right)\\
&\quad\leq
e^{-\beta t}H(\omega\|\rho\otimes\omega_{\mathsf k}).
\end{split}
\]
The reference state remains since $\mathcal R_t^*(\rho)=\rho$.
\end{proof}

\begin{lemma}\label{lem:section4-relative-gns-rate}
Under the assumptions of Lemma~\ref{lem:section4-fixed-splitting},
suppose $l<d$ and set
\[
\beta=
\lambda_{\min}\!\left(
\mathbf F_{\mathrm r}^{-1/2}
\bfC_{\bfZ,\mathrm r}
\mathbf F_{\mathrm r}^{-1/2}
\right).
\]
Then
\[
\lambda_{\mathrm{GNS}}(\Phi;\mathbb E_\Phi)=\frac{\beta}{2},
\]
where $\lambda_{\mathrm{GNS}}$ is the optimal GNS norm-decay rate
with prefactor one. Moreover,
\[
\beta
=\sup\left\{b\geq0:
\bfC_{\bfZ}\geq
b(\mathbf F-\mathbf P^*\mathbf F\mathbf P)\right\}=\inf_{\mathbf z\notin\ker_{\mathbb C}\mathbf Z}
\frac{\mathbf z^*\bfC_{\bfZ}\mathbf z}
{\operatorname{dist}_{\mathbf F}
(\mathbf z,\ker_{\mathbb C}\mathbf Z)^2}.
\]
\[
\beta>0
\quad \text{ if and only if }\quad
\ker_{\mathbb C}\bfC_{\bfZ}
=\ker_{\mathbb C}\mathbf Z.
\]
The rate identities also hold when $\beta=0$; no symmetry assumption
is required.
\end{lemma}
\begin{proof}
In the adapted coordinates of Lemma~\ref{lem:section4-fixed-splitting},
\[
\mathbf P=I_{2l}\oplus0,
\qquad
\bfC_{\bfZ}=0\oplus\bfC_{\bfZ,\mathrm r},
\qquad
\mathbf F-\mathbf P^*\mathbf F\mathbf P=0\oplus\mathbf F_{\mathrm r}.
\]
Since $\mathbf Z_{\mathrm r}$ is invertible and
$\mathbf P^*\mathbf F=\mathbf F\mathbf P$, the matrix $\mathbf P$
is the $\mathbf F$-orthogonal projection onto
$\ker_{\mathbb C}\mathbf Z=\mathbb C^{2l}\oplus\{0\}$.
Thus the Rayleigh principle gives
\[
\begin{aligned}
\beta
&=\inf_{\mathbf z_{\mathrm r}\neq0}
\frac{\mathbf z_{\mathrm r}^*\bfC_{\bfZ,\mathrm r}\mathbf z_{\mathrm r}}
{\mathbf z_{\mathrm r}^*\mathbf F_{\mathrm r}\mathbf z_{\mathrm r}}\\
&=\inf_{\mathbf z\notin\ker_{\mathbb C}\mathbf Z}
\frac{\mathbf z^*\bfC_{\bfZ}\mathbf z}
{\operatorname{dist}_{\mathbf F}(\mathbf z,\ker_{\mathbb C}\mathbf Z)^2},
\end{aligned}
\]
and, for $b\geq0$,
\[
\bfC_{\bfZ}\geq b(\mathbf F-\mathbf P^*\mathbf F\mathbf P)
\iff\bfC_{\bfZ,\mathrm r}\geq b\mathbf F_{\mathrm r}
\iff b\leq\beta.
\]
The same block decomposition and $\mathbf F_{\mathrm r}>0$ imply
\[
\beta>0\iff\bfC_{\bfZ,\mathrm r}>0
\iff\ker_{\mathbb C}\bfC_{\bfZ}=\ker_{\mathbb C}\mathbf Z.
\]
These characterizations are invariant under the coordinate change:
the three quadratic forms transform by the same congruence, and
the kernel and its $\mathbf F$-distance transform accordingly.

By \eqref{eq: GNS isometry} and the totality argument of
Proposition~\ref{prop:second}, the map
$U_{\mathrm r}\widetilde W(\mathbf z)=e_{\mathbf F_{\mathrm r}^{1/2}\mathbf z}$
for $\mathbf z\in\mathbb R^{2(d-l)}$
extends to a unitary from $L_{2,\mathrm{GNS}}(\rho_{\mathrm r})$
onto $\cF_s(\mathbb C^{2(d-l)})$.
The Weyl intertwining relation then gives
\[
U_{\mathrm r}\Phi_{t,\mathrm r}^{\mathrm{GNS}}U_{\mathrm r}^*
=\bigoplus_{n=0}^{\infty}
\left(\mathbf F_{\mathrm r}^{1/2}e^{t\mathbf Z_{\mathrm r}}
\mathbf F_{\mathrm r}^{-1/2}\right)^{\otimes_s n},
\]
where $\Phi_{t,\mathrm r}^{\mathrm{GNS}}$ is the GNS extension.
It is contractive because $\Phi_{t,\mathrm r}$ is unital, completely
positive, and preserves $\rho_{\mathrm r}$; hence its first-chaos
restriction is also contractive.
Since $U_{\mathrm r}1_{\mathrm r}=e_0$, it follows that
\[
\left\|\Phi_{t,\mathrm r}^{\mathrm{GNS}}
\big|_{\{1_{\mathrm r}\}^{\perp}}\right\|
=\sup_{n\geq1}
\left\|\mathbf F_{\mathrm r}^{1/2}e^{t\mathbf Z_{\mathrm r}}
\mathbf F_{\mathrm r}^{-1/2}\right\|_{\mathrm{op}}^n
=\left\|\mathbf F_{\mathrm r}^{1/2}e^{t\mathbf Z_{\mathrm r}}
\mathbf F_{\mathrm r}^{-1/2}\right\|_{\mathrm{op}}.
\]
By the product decomposition in Lemma~\ref{lem:section4-fixed-splitting},
\[
\ker\bE_\Phi=L_{2,\mathrm{GNS}}(\rho_0)\otimes
\{1_{\mathrm r}\}^{\perp}_{L_{2,\mathrm{GNS}}(\rho_{\mathrm r})}.
\]
Here $\bE_\Phi$ denotes its orthogonal projection extension to
$L_{2,\mathrm{GNS}}(\rho)$.
As $\Phi_t=\Id\otimes\Phi_{t,\mathrm r}$, its GNS extension
has exactly the centered relaxing norm on this kernel.

Finally, matrix differentiation gives
\[
\frac{d}{dt}\left(e^{t\mathbf Z_{\mathrm r}^*}\mathbf F_{\mathrm r}
e^{t\mathbf Z_{\mathrm r}}\right)
=-e^{t\mathbf Z_{\mathrm r}^*}\bfC_{\bfZ,\mathrm r}
e^{t\mathbf Z_{\mathrm r}}.
\]
For $b\geq0$, integrating this identity with an integrating factor,
and conversely differentiating the resulting inequality at $t=0$,
shows that
\[
\bfC_{\bfZ,\mathrm r}\geq b\mathbf F_{\mathrm r}
\iff
e^{t\mathbf Z_{\mathrm r}^*}\mathbf F_{\mathrm r}e^{t\mathbf Z_{\mathrm r}}
\leq e^{-bt}\mathbf F_{\mathrm r}\quad(t\geq0).
\]
The latter is equivalent to
$\|\mathbf F_{\mathrm r}^{1/2}e^{t\mathbf Z_{\mathrm r}}
\mathbf F_{\mathrm r}^{-1/2}\|_{\mathrm{op}}\leq e^{-bt/2}$
for all $t\geq0$.
The order characterization of $\beta$ therefore proves
$\lambda_{\mathrm{GNS}}(\Phi;\mathbb E_\Phi)=\beta/2$,
including optimality when $\beta=0$.
\end{proof}

\begin{lemma}[Complete entropy contraction implies GNS contraction]
\label{lem:section4-entropy-implies-gns}
Let $\rho$ be a faithful density operator on a separable Hilbert
space $\mathsf h$, and let
$\Phi:\mathcal B(\mathsf h)\to\mathcal B(\mathsf h)$ be a normal
unital completely positive map satisfying $\Phi^*(\rho)=\rho$.
Suppose that, for some $\eta\in[0,1]$ and $\mathsf k=\mathbb C^2$,
\[
H\!\left(
(\Phi^*\otimes\Id_{\mathsf k})(\omega)
\,\middle\|\,
\rho\otimes\omega_{\mathsf k}
\right)
\leq
\eta H(\omega\|\rho\otimes\omega_{\mathsf k})
\]
for every joint density $\omega$ for which the relative entropy
on the right is finite. Then
\[
\left\|
\Phi^{\mathrm{GNS}}\big|_{\{I\}^{\perp}}
\right\|^2\leq\eta.
\]
\end{lemma}

\begin{proof}
Choose an orthonormal eigenbasis $\rho f_i=r_i f_i$, with
$r_i>0$, and let $P_K$ be the projection onto
$\operatorname{span}\{f_1,\ldots,f_K\}$.
In finite dimension, take $P_K=I$ and omit the compression below.

Fix an operator $A$ with finite matrix support in this basis
and $\operatorname{Tr}A=0$.
For $q\in(0,1)$ and sufficiently small real $\varepsilon$,
\[
\omega=
\begin{pmatrix}
(1-q)\rho&\varepsilon A\\
\varepsilon A^*&q\rho
\end{pmatrix}
\]
is a density operator with marginals
$\omega_{\mathsf h}=\rho$ and
$\omega_{\mathsf k}=\operatorname{diag}(1-q,q)$.
For every sufficiently large $K$ such that $A=P_KAP_K$,
the states $\omega$ and $\rho\otimes\omega_{\mathsf k}$ are block diagonal
with respect to
\[
 (P_K\mathsf h\otimes\mathbb C^2)
 \oplus((I-P_K)\mathsf h\otimes\mathbb C^2).
\]
They agree on the second summand, and the reference state is
strictly positive on the finite-dimensional first summand.
Hence $H(\omega\|\rho\otimes\omega_{\mathsf k})<\infty$.

For every such $K$, consider the channel
\[
\Lambda_K(S)
=
P_KSP_K\oplus\operatorname{Tr}((I-P_K)S).
\]
Compression leaves the input relative entropy unchanged.
By data processing and the assumed contraction,
\[
H\!\left(
((\Lambda_K\circ\Phi^*)\otimes\Id_{\mathsf k})(\omega)
\,\middle\|\,
\Lambda_K(\rho)\otimes\omega_{\mathsf k}
\right)
\leq
\eta H(\omega\|\rho\otimes\omega_{\mathsf k}).
\]
Applying the finite-dimensional second-order expansion of relative
entropy, as in the proof of \cite[Lemma~2.3]{LWW26}, to the compressed
input and output gives the following inequality after division by
$\varepsilon^2$ and passage to $\varepsilon=0$.
The added output coordinate contributes a nonnegative quadratic term,
which we discard:
\[
\begin{aligned}
\sum_{i,j\leq K}
|(\Phi^*(A))_{ij}|^2
\frac{\log((1-q)r_i)-\log(qr_j)}
     {(1-q)r_i-qr_j}\leq
\eta\sum_{i,j}|A_{ij}|^2
\frac{\log((1-q)r_i)-\log(qr_j)}
     {(1-q)r_i-qr_j}.
\end{aligned}
\]
The quotients are understood by continuity when their
denominators vanish.

Since
\[
\lim_{q\downarrow0}
\frac{\log((1-q)r_i)-\log(qr_j)}
     {((1-q)r_i-qr_j)\log(1/q)}
=
\frac1{r_i},
\]
divide by $\log(1/q)$ and let $q\downarrow0$ for fixed $K$.
Then let $K\to\infty$ by monotone convergence, omitting this
last limit in finite dimension.
We obtain
\begin{equation}
\label{eq:section4-weighted-predual-contraction}
\begin{aligned}
\|\rho^{-1/2}\Phi^*(A)\|_2^2
&=
\sum_{i,j}\frac{|(\Phi^*(A))_{ij}|^2}{r_i}\leq
\eta\sum_{i,j}\frac{|A_{ij}|^2}{r_i}
=
\eta\|\rho^{-1/2}A\|_2^2.
\end{aligned}
\end{equation}
Here $\|\cdot\|_2$ is the Hilbert--Schmidt norm, and the
weighted expressions are interpreted through the displayed sums.

For $X\in\mathcal B(\mathsf h)$ with $\operatorname{Tr}(\rho X)=0$,
Hilbert--Schmidt duality gives
\[
\|X\|_{\mathrm{GNS},\rho}
=
\sup_{\substack{
A\text{ has finite matrix support}\\
\operatorname{Tr}A=0,\;
\|\rho^{-1/2}A\|_2\leq1}}
|\operatorname{Tr}(AX)|.
\]
Indeed, Cauchy--Schwarz gives the upper bound.
Without the finite-support restriction, equality is attained
at $A=\rho X^*/\|X\|_{\mathrm{GNS},\rho}$ when $X\neq0$.
For this $A$, the operators
\[
A_K
=
P_KAP_K-\operatorname{Tr}(P_KAP_K)
|f_1\rangle\langle f_1|
\]
have finite matrix support and trace zero, and satisfy
$\|\rho^{-1/2}(A_K-A)\|_2\to0$.
Rescaling them if necessary preserves the unit-norm constraint,
so the finite-support restriction does not change the supremum.

Since $\operatorname{Tr}(\rho\Phi(X))=0$, the same duality and
\eqref{eq:section4-weighted-predual-contraction} yield
\[
\begin{aligned}
\|\Phi(X)\|_{\mathrm{GNS},\rho}
&=
\sup_A|\operatorname{Tr}(A\Phi(X))|\\
&=
\sup_A|\operatorname{Tr}(\Phi^*(A)X)|\\
&\leq
\sup_A
\|\rho^{-1/2}\Phi^*(A)\|_2
\|X\rho^{1/2}\|_2\\
&\leq
\sqrt{\eta}\,\|X\|_{\mathrm{GNS},\rho}.
\end{aligned}
\]
The operator-norm conclusion follows by density of the centered
bounded operators in $\{I\}^{\perp}$.
\end{proof}

\begin{proof}[Proof of Theorem~\ref{thm:section4-relative-cmlsi}]
Let $\beta$ be as in Lemma~\ref{lem:section4-relative-gns-rate}.
The decomposition in Lemma~\ref{lem:section4-fixed-splitting}
and Gaussian comparison in Lemma~\ref{lem:section4-gaussian-comparison},
applied to the relaxing factor with auxiliary space
$\mathsf h_0\otimes\mathbb C^k$, give relative entropy contraction
at rate $\beta$. Indeed, after reordering tensor factors, the
reference density is
$(\mathbb E_\Phi^*\otimes\Id_k)(\omega)$.
Hence $\alpha_{\mathrm{CMLSI}}(\Phi;\mathbb E_\Phi)\geq\beta$.

Conversely, any admissible relative CMLSI rate $a\geq0$ applies
to $\rho_0\otimes\omega$, with $\omega$ a joint density on
$\mathsf h_{\mathrm r}\otimes\mathbb C^2$ and finite relative entropy.
Cancelling the common factor $\rho_0$ gives entropy contraction
for $\Phi_{t,\mathrm r}$ with two-dimensional auxiliary space
and factor $e^{-at}$.
Lemma~\ref{lem:section4-entropy-implies-gns} therefore gives
centered GNS norm contraction at rate $a/2$.
By Lemma~\ref{lem:section4-relative-gns-rate}, its optimal rate
is $\beta/2$, so $a\leq\beta$.
Taking the supremum over $a$ proves
\eqref{eq:section4-cmlsi-bound}; the matrix characterizations
and strict positivity criterion are exactly those of
Lemma~\ref{lem:section4-relative-gns-rate}.
\end{proof}

The same comparison gives the exact complete entropy contraction coefficient at each fixed time.
\begin{proposition}\label{prop:section4-fixed-time}
Under the assumptions of Lemma~\ref{lem:section4-fixed-splitting}, suppose that $l<d$.
For $t\geq0$, let $\eta_{\mathrm c}(\Phi_t;\mathbb E_\Phi)$ be the smallest $\eta\in[0,1]$ for which \eqref{eq:section4-cmlsi} holds at this time with $e^{-\beta t}$ replaced by $\eta$, for every finite auxiliary dimension and every input with finite relative entropy.
Then
\begin{equation}\label{eq:section4-fixed-time}
 \eta_{\mathrm c}(\Phi_t;\mathbb E_\Phi)
 =\left\|\mathbf F_{\mathrm r}^{1/2}e^{t\mathbf Z_{\mathrm r}}
 \mathbf F_{\mathrm r}^{-1/2}\right\|_{\mathrm{op}}^2.
\end{equation}
In particular, $\eta_{\mathrm c}(\Phi_t;\mathbb E_\Phi)<1$ for every $t>0$, even when $\alpha_{\mathrm{CMLSI}}=0$.
\end{proposition}
\begin{proof}
Fix $t\geq0$ and denote the right side of \eqref{eq:section4-fixed-time} by $c$.
The matrix exponential is invertible and contractive in the $\mathbf F_{\mathrm r}$ norm, so $0<c\leq1$.
Put $\mathbf S=c^{-1/2}e^{t\mathbf Z_{\mathrm r}}$. Then
\[
 \mathbf F_{\mathrm r}-\mathbf S^*\mathbf F_{\mathrm r}\mathbf S\geq0.
\]
The Gaussian complete-positivity criterion therefore gives a normal $\rho_{\mathrm r}$-preserving channel $\mathcal R$ with Weyl action
\[
 \mathcal R(W(\mathbf z))
 =\exp\!\left[-\frac12\mathbf z^{\mathsf T}
 (\Sigma_{\mathrm r}-\mathbf S^*\Sigma_{\mathrm r}\mathbf S)\mathbf z\right]
 W(\mathbf S\mathbf z).
\]
Direct composition with the reference semigroup of Lemma~\ref{lem:section4-reference-entropy}, for covariance $\Sigma_{\mathrm r}$, gives
\[
 \Phi_{t,\mathrm r}=\mathcal R\circ\Psi_{-\log c}
 =\Psi_{-\log c}\circ\mathcal R.
\]
That lemma and data processing give entropy contraction by $c$ on the relaxing factor, with any separable auxiliary system.
Taking the fixed sector together with $\mathbb C^k$ as the auxiliary system proves
$\eta_{\mathrm c}(\Phi_t;\mathbb E_\Phi)\leq c$.
Here $\mathcal R$ is a channel at the fixed time $t$; no semigroup property for these channels is needed.

Conversely, any admissible coefficient $\eta$ applies to inputs $\rho_0\otimes\omega$, where $\omega$ is a joint density on the relaxing factor and $\mathbb C^2$.
Cancelling the fixed factor and applying Lemma~\ref{lem:section4-entropy-implies-gns} gives
\[
 \eta\geq
 \left\|\Phi_{t,\mathrm r}^{\mathrm{GNS}}
 \big|_{\{1_{\mathrm r}\}^{\perp}}\right\|^2=c,
\]
where the equality follows from the second-quantization representation in Lemma~\ref{lem:section4-relative-gns-rate}.
This proves \eqref{eq:section4-fixed-time}.
Since $\mathbf Z_{\mathrm r}$ is Hurwitz, Proposition~\ref{prop:control} and \eqref{eq:Fidentity}, applied to the relaxing factor, give $c<1$ for $t>0$.
\end{proof}

\begin{corollary}\label{cor:section4-kernel-equilibration}
Suppose that $\{\Phi_t\}_{t\geq0}$ is a finite-mode Gaussian
quantum Markov semigroup with a faithful centered invariant
Gaussian density.
Then the normal $\rho$-preserving fixed-point conditional expectation $\mathbb E_\Phi$ exists, and
\[
\alpha_{\mathrm{CMLSI}}(\Phi;\mathbb E_\Phi)>0
\quad\text{if and only if}\quad
\ker_{\mathbb C}\bfC_{\bfZ}
=
\ker_{\mathbb C}\mathbf Z.
\]
Whenever these equivalent conditions hold,
$\displaystyle\lim_{t\to\infty}e^{t\mathbf Z}$ exists.
If the semigroup is not the identity, its optimal rate is given by
\eqref{eq:section4-cmlsi-bound} and is strictly positive.

If all modes are fixed, then $\Phi_t=\Id$ and
$\mathbb E_\Phi=\Id$; the relative entropy is identically
zero, and $\alpha_{\mathrm{CMLSI}}=+\infty$.
\end{corollary}

\begin{proof}
The Gaussian decomposition and the trace-norm Ces\`aro convergence in
\cite[Theorems~11 and~20]{GP26} give a normal conditional expectation onto the fixed-point algebra.
It preserves $\rho$ because each $\Phi_t$ does.
Assume first that $\ker_{\mathbb C}\bfC_{\bfZ}=\ker_{\mathbb C}\mathbf Z$.
Since $\bfC_{\bfZ}\geq0$, stationarity gives
$e^{t\mathbf Z^*}\mathbf F e^{t\mathbf Z}\leq\mathbf F,
\; t\geq0.$
As $\mathbf F>0$, the semigroup $e^{t\mathbf Z}$ is uniformly
bounded. Hence every eigenvalue of $\mathbf Z$ has nonpositive
real part, and every eigenvalue on the imaginary axis is
semisimple.

If $\mathbf Z\mathbf z=i\omega\mathbf z$ for
$\omega\in\mathbb R$ and $\mathbf z\neq0$, then
$\mathbf z^*\bfC_{\bfZ}\mathbf z
=
-2\operatorname{Re}
\bigl(\mathbf z^*\mathbf F\mathbf Z\mathbf z\bigr)
=0$.
Positivity and the kernel assumption imply
$\bfC_{\bfZ}\mathbf z=0
\Longrightarrow
\mathbf Z\mathbf z=0,$
so $\omega=0$.
Thus every nonzero eigenvalue has strictly negative real part,
while zero, if present, is semisimple.
Consequently, $\lim\limits_{t\to\infty}e^{t\mathbf Z}$ exists.

If the relaxing sector is nonempty,
Theorem~\ref{thm:section4-relative-cmlsi} now gives
\[
\alpha_{\mathrm{CMLSI}}(\Phi;\mathbb E_\Phi)
=
\beta
=
\lambda_{\min}\!\left(
\mathbf F_{\mathrm r}^{-1/2}
\bfC_{\bfZ,\mathrm r}
\mathbf F_{\mathrm r}^{-1/2}
\right)>0.
\]
If all modes are fixed, then $\mathbf Z=0$ and stationarity
gives $\mathbf C=0$.
The Weyl formula yields
$\Phi_t=\mathbb E_\Phi=\Id$, so the relative entropy
vanishes and the contraction estimate holds for every
$\beta>0$.

Conversely, suppose that $\alpha_{\mathrm{CMLSI}}(\Phi;\mathbb E_\Phi)>0$.
If $e^{t\mathbf Z}$ does not converge, the decomposition in
\cite[Theorem~11]{GP26} contains a noiseless factor with Hamiltonian
$H_0=\sum_j\omega_j\bfa_j^\dagger\bfa_j$, where some $\omega_j\neq0$,
tensored with a relaxing Gaussian semigroup with invariant density $\rho_{\mathrm r}$.
By \cite[Theorem~20]{GP26}, the fixed-point expectation on this factor is the pinching onto the distinct energy eigenspaces of $H_0$.
Choose unit eigenvectors $u,v$ with different energies; one may take the vacuum and a one-particle vector in a mode with nonzero frequency.
For
\[
 \omega=\frac12|u+v\rangle\langle u+v|\otimes\rho_{\mathrm r},
 \qquad
 \mathbb E_\Phi^*(\omega)
 =\frac12\bigl(|u\rangle\langle u|+|v\rangle\langle v|\bigr)
 \otimes\rho_{\mathrm r},
\]
the noiseless evolution changes only the relative phase of $u$ and $v$.
Thus
\[
 H\!\left(\Phi_t^*(\omega)\,\middle\|\,\mathbb E_\Phi^*(\omega)\right)
 =\log2
 =H\!\left(\omega\,\middle\|\,\mathbb E_\Phi^*(\omega)\right)
 \qquad(t\geq0).
\]
If the relaxing factor is empty, its tensor factor is omitted.
This finite-entropy input contradicts every strictly positive rate in \eqref{eq:section4-cmlsi}, already for $k=1$.
Consequently, $e^{t\mathbf Z}$ converges.
If the semigroup is not the identity, Theorem~\ref{thm:section4-relative-cmlsi} now gives the kernel equality.
For the identity semigroup, $\mathbf Z=\mathbf C=0$, so the same equality holds.
\end{proof}

\begin{remark}\label{rem:section4-relative-scope}
Under the assumptions of
Theorem~\ref{thm:section4-relative-cmlsi},
the optimal rate cannot in general be replaced by the smallest
strictly positive eigenvalue of
$\mathbf F^{-1/2}\bfC_{\bfZ}\mathbf F^{-1/2}$.
Indeed,
\[
\ker_{\mathbb C}\bfC_{\bfZ}
\neq\ker_{\mathbb C}\mathbf Z
\quad\Longrightarrow\quad
\alpha_{\mathrm{CMLSI}}(\Phi;\mathbb E_\Phi)=0,
\]
even if this matrix has strictly positive eigenvalues.

The fixed-sector decomposition uses the actual limit
$\mathbf P=\lim_{t\to\infty}e^{t\mathbf Z}$.
Replacing it by a Ces\`aro limit does not, in general,
preserve the identification of the fixed-point algebra and
its conditional expectation, since persistent noiseless
rotations may then occur.
\end{remark}

The following two examples illustrate the distinction between
KMS norm decay and complete entropy contraction.
The first shows that a positive KMS norm-decay rate does not
guarantee a positive CMLSI constant, even for a KMS-symmetric
Gaussian semigroup.
The second exhibits a family of Gaussian semigroups with a
strictly positive optimal CMLSI rate, although KMS symmetry
fails when the rotation parameter is nonzero.
In this family, the optimal rate is unchanged by the added
Hamiltonian, which commutes with the invariant density.
\begin{example}\label{ex:anisotropic-boundary}
Let
\begin{equation}\label{eq:anisotropic-example}
\Sigma=\nu I_2,\qquad
\mathbf Z=-\operatorname{diag}(r_1,r_2),\qquad
\mathbf C=2\nu\operatorname{diag}(r_1,r_2),
\qquad
\nu>1,\quad r_1,r_2>0.
\end{equation}
Then
\[
\mathbf F=
\begin{pmatrix}
\nu&i\\
-i&\nu
\end{pmatrix}>0,
\]
and the quantum diffusion matrix is
\begin{equation}\label{eq:anisotropic-Q}
\bfC_{\bfZ}
=
\begin{pmatrix}
2\nu r_1&i(r_1+r_2)\\
-i(r_1+r_2)&2\nu r_2
\end{pmatrix}.
\end{equation}
Thus these coefficients define a Gaussian quantum Markov
semigroup if and only if
\begin{equation}\label{eq:anisotropic-cp}
4\nu^2r_1r_2\geq(r_1+r_2)^2.
\end{equation}
Assume that \eqref{eq:anisotropic-cp} holds.
Since $\mathbf Z$ is Hurwitz, the semigroup has a unique
faithful centered invariant Gaussian density $\rho$ with
covariance $\Sigma$.

The semigroup is KMS-symmetric, since
\[
\Sigma_{\mathrm{KMS}}=\sqrt{\nu^2-1}\,I_2,
\qquad
\mathbf Z=\mathbf Z^{\mathsf T}.
\]
It is GNS-symmetric if and only if $r_1=r_2$.
The optimal KMS norm-decay rate is $\min(r_1,r_2)$;
equivalently,
\[
\gamma_{\mathrm{KMS}}=2\min(r_1,r_2).
\]
By Theorem~\ref{thm:section4-relative-cmlsi}, the exact CMLSI
constant is
\begin{equation}\label{eq:anisotropic-rate}
\alpha_{\mathrm{CMLSI}}(\Phi;\mathbb E_\Phi)
=
r_1+r_2-
\frac{\nu}{\sqrt{\nu^2-1}}|r_1-r_2|.
\end{equation}
Indeed,
\[
\begin{aligned}
\det(\bfC_{\bfZ}-c\mathbf F)
&=
\nu^2(2r_1-c)(2r_2-c)-(r_1+r_2-c)^2\\
&=
(\nu^2-1)(r_1+r_2-c)^2
-\nu^2(r_1-r_2)^2.
\end{aligned}
\]
The smaller root is \eqref{eq:anisotropic-rate}, and
\eqref{eq:anisotropic-cp} is equivalent to its nonnegativity.
In the isotropic case $r_1=r_2$, the optimal CMLSI constant
is $2r_1$.
At the nonisotropic completely positive boundary, it is zero
despite a strictly positive KMS norm-decay rate.
\end{example}

\begin{example}
Take
\begin{equation}\label{eq:rotating-example}
\Sigma=2I_2,\qquad
\mathbf Z_\omega=
\begin{pmatrix}
-1&\omega\\
-\omega&-2
\end{pmatrix},\qquad
\mathbf C=
\begin{pmatrix}
4&0\\
0&8
\end{pmatrix},
\qquad
\omega\in\mathbb R.
\end{equation}
Stationarity follows from
\[
\mathbf Z_\omega^{\mathsf T}\Sigma
+\Sigma\mathbf Z_\omega+\mathbf C=0.
\]
Moreover,
\[
\mathbf Z_\omega+\mathbf Z_\omega^{\mathsf T}
=
-2\operatorname{diag}(1,2)<0,
\]
so $\mathbf Z_\omega$ is Hurwitz.
For every $\omega\in\mathbb R$,
\[
\mathbf F=
\begin{pmatrix}
2&i\\
-i&2
\end{pmatrix},
\qquad
\bfC_{\bfZ}=
\begin{pmatrix}
4&3i\\
-3i&8
\end{pmatrix}>0.
\]

Writing $a=a_1$, $Q=Q_1$, and $P=P_1$, a concrete GKSL
representation is
\begin{equation}\label{eq:rotating-gksl}
H=-\frac14(QP+PQ)-\omega a^\dagger a,\qquad
L_1=\frac34Q+2iP,\qquad
L_2=\frac{\sqrt{23}}4Q.
\end{equation}
The Lindblad operators give the dissipative Weyl-label drift
$-3I_2/2$ and diffusion $\operatorname{diag}(4,8)$.
The Hamiltonian contributes the Weyl-label drift $
\operatorname{diag}(1/2,-1/2)+\omega\mathbf J$,
yielding the coefficients in \eqref{eq:rotating-example}.

The unique invariant density is
$\rho=\frac23\,3^{-a^\dagger a}$,
whose covariance is
$\coth((\log3)/2)I_2=2I_2$.
Since
$\Sigma_{\mathrm{KMS}}=\sqrt3\,I_2$,
the semigroup is KMS-symmetric if and only if $\omega=0$.

A direct computation gives
$\det(\bfC_{\bfZ}-c\mathbf F)
=
3(3-c)^2-4$.
Hence Theorem~\ref{thm:section4-relative-cmlsi} gives, for every
$\omega\in\mathbb R$,
\begin{equation}\label{eq:rotating-rate}
\alpha_{\mathrm{CMLSI}}(\Phi;\mathbb E_\Phi)
=
3-\frac2{\sqrt3}
\approx1.8452994616.
\end{equation}
The additional Hamiltonian $-\omega a^\dagger a$ commutes
with $\rho$, so it contributes no relative entropy production.
Thus the optimal rate is independent of $\omega$, even though
the rotation need not commute with the anisotropic damping.
\end{example}

\section{Hypercontractivity}
In this section, we shall obtain hypercontractivity for thermal Gaussian quantum Markov semigroup and then obtain the hypercontractivity for general one.

By taking the drift $\bfZ=-I_{2d}/2$ and the diffusion $\bfC=\Sigma$, we define a quantum Markov semigroup $\Psi_s$ as follows:
\begin{equation}\label{eq:thermal-ref}
 \Psi_s(W(\mathbf z))=
 \exp\left(-\frac12(1-e^{-s})\mathbf z^{\mathsf{T}}\Sigma\mathbf z\right)
 W(e^{-s/2}\mathbf z).
\end{equation}
Then it is a GNS-symmetric Gaussian quantum Markov semigroup with faithful invariant state $\rho$. 
In Williamson coordinates it is the product of one-mode thermal semigroups with the same covariance-relaxation rate one, but possibly different temperatures.

For the one-mode case and a general rate $\gamma>0$, the Lindbladian operator is
\begin{equation}\label{eq:thermal-L}
 -\cL_{\gamma,\beta}(X)
 =\gamma\frac{e^{\beta}}{e^\beta-1} \left(\bfa^\dag X\bfa-\frac{1}{2}\{\bfa^*\bfa,X\}\right)
 +\gamma \frac{1}{e^\beta-1} \left(\bfa X\bfa^\dag-\frac{1}{2}\{\bfa\bfa^{\dag},X\}\right).
\end{equation}
The corresponding matrices $\bfZ, \Sigma, \bfC$ are 
\begin{align*}
 \bfZ=-\frac\gamma2I_2,\qquad
 \Sigma=\coth(\beta/2)I_2,\qquad \bfC=\gamma\Sigma.
\end{align*}
In particular, $e^{-t\cL_{\gamma,\beta}}(\bfa)=e^{-\gamma t/2}\bfa$, understood in the corresponding $L_2(\rho_\beta)$ realization.

In the following, we shall obtain the hypercontractivity of the thermal Gaussian quantum Markov semigroup.
\begin{proposition}\label{thm:thermal}
Suppose that $\{\Psi_s\}_{s\geq 0}$ is a Gaussian quantum Markov semigroup given in \eqref{eq:thermal-ref}.
Then there is an explicit constant $\hc>0$ such that, for all $1<p\leq q<\infty$,
\begin{equation}\label{eq:thermal-hc}
\|\Psi_s\|_{p\to q,\rho}\leq 1
\end{equation}
whenever $\displaystyle s\geq\hc^{-1}\log\frac{q-1}{p-1}.$
\end{proposition}
\begin{proof}
For one mode, rescale the generator in \cite[Theorem 4 and Corollary 5]{BRK25} to covariance relaxation rate one. This gives the constant in \eqref{eq:h-one}.

For the multi-mode case, let $\cL_j$ denote the positive generator of the $j$th unit-rate thermal factor, so that $\Psi_s=e^{-s\sum_j\cL_j}$.
In the Heisenberg picture, the generator in \cite[Eq.~(38) and Theorem 10]{BRK25} is
\[
 \widehat{\cL}=\sum_{j=1}^d2\sinh(\beta_j/2)\cL_j,
\]
and its $2$-logarithmic Sobolev constant is at least $a_d$ from \eqref{eq:ahat}.
Since each mode has a nonnegative Dirichlet form and $\beta_+=\max_j\beta_j$,
\[
 \mathcal E_{2,\widehat{\cL}}(X)
 =\sum_{j=1}^d2\sinh(\beta_j/2)\mathcal E_{2,\cL_j}(X)
 \leq2\sinh(\beta_+/2)\mathcal E_{2,\sum_j\cL_j}(X).
\]
Consequently,
\[
 \alpha_2\Bigl(\sum_j\cL_j\Bigr)
 \geq\frac{a_d}{2\sinh(\beta_+/2)}.
\]
The log Sobolev--hypercontractivity implication in \cite{OZ99,BDR20}, together with GNS symmetry, therefore gives the unit-rate hypercontractivity constant
$\hc=2a_d/\sinh(\beta_+/2)$ in \eqref{eq:h-multi}.
This completes the proof of the proposition.
\end{proof}

\begin{remark}
In the one-mode case, the explicit constant $\hc$ can be taken as
\begin{equation}\label{eq:h-one}
 \hc=\frac4\beta\tanh(\beta/4).
\end{equation}
The generator in \cite[Theorem 4 and Corollary 5]{BRK25} has covariance relaxation rate $2\sinh(\beta/2)$; dividing by this rate gives \eqref{eq:h-one}.
\end{remark}

\begin{remark}
For the multi-mode case, we follow the notations from \cite[Theorem 10]{BRK25} and take
\begin{equation}\label{eq:ahat}
 a_d=\left[
 \frac{2+\log(2d+1)}{\sinh(\beta_-/2)}
 +\frac{\beta_-}{4\sinh^2(\beta_-/4)}
 \right]^{-1},
\end{equation}
where $\displaystyle \beta_-=\min_j\beta_j$, $\displaystyle \beta_+=\max_j\beta_j$.
A possible choice for $\hc$ is
\begin{equation}\label{eq:h-multi}
 \hc=\frac{2a_d}{\sinh(\beta_+/2)}.
\end{equation}
When $\beta_j$ are equal to $\beta$, we see that 
\begin{equation}\label{eq:h-equal}
 \hc=\frac4{2(2+\log(2d+1))+\beta\coth(\beta/4)}.
\end{equation}
\end{remark}

In the following, we shall obtain the exact one-mode thermal hypercontractivity time.
A constant-rate curve need not give the best sufficient time for a fixed pair of exponents. Define
\begin{equation}\label{eq:f-beta}
 f_\beta(r)=
 \frac{\sinh\bigl(\beta(r-1)/(2r)\bigr)}{\sinh\bigl(\beta/(2r)\bigr)},
 \qquad r>1.
\end{equation}
It is increasing and $f_\beta(2)=1$.
\begin{proposition}\label{prop:thermal-profile}
For the thermal semigroup $\Phi_t=e^{-t\cL_{\gamma,\beta}}$
in \eqref{eq:thermal-L}, and $1<p\leq q<\infty$, we have
\begin{equation}\label{eq:thermal-profile}
 \norm{\Phi_t}_{p\to q,\rho_\beta}
 =\begin{cases}
 +\infty,&\displaystyle 0\leq t<\frac1\gamma
       \log\frac{f_\beta(q)}{f_\beta(p)},\\[4pt]
 1,&\displaystyle t\geq\frac1\gamma
       \log\frac{f_\beta(q)}{f_\beta(p)}.
 \end{cases}
\end{equation}
\end{proposition}
\begin{proof}
In the covariance-rate normalization of \eqref{eq:thermal-L},
the constants of \cite[Theorem~4]{BRK25}, valid for every $r>1$, are
\begin{equation}\label{eq:alpha-thermal}
 \alpha_r=\frac{\gamma r^2}{4\beta(r-1)}
 \frac{(1-e^{-\beta/r})(1-e^{-\beta(r-1)/r})}{1-e^{-\beta}}.
\end{equation}
Integrating the logarithmic Sobolev differentiation identity with
$r'(t)=4(r(t)-1)\alpha_{r(t)}$ gives the sufficient time
\[
 \int_p^q\frac{dr}{4(r-1)\alpha_r}
 =\frac1\gamma\log\frac{f_\beta(q)}{f_\beta(p)},
\]
where the equality follows by differentiating \eqref{eq:f-beta}.
The argument is first applied to positive elements of the thermal
generator core, regularized by a positive multiple of $1$.
The Feller property and compact-operator approximation used in
\cite[proof of Corollary~5]{BRK25} extend the estimate to the
completed weighted $L_p$ space. Complete positivity allows the
positive-input estimate to control the full $p\to q$ norm.
Since $\Phi_t(1)=1$ and $\norm{1}_{r,\rho_\beta}=1$, the norm
is exactly one at and after this time.

For necessity, suppose that $p<q$ and
$0\leq t<\gamma^{-1}\log(f_\beta(q)/f_\beta(p))$.
Put $b=e^\beta$ and $\eta=e^{-\gamma t}$.
For $1<z<b^{1/p}$, consider the bounded finite-rank operators
\[
 X_{z,M}=\sum_{n=0}^M z^n|n\rangle\langle n|.
\]
Their input norms satisfy
\[
 \norm{X_{z,M}}_{p,\rho_\beta}^p
 =(1-b^{-1})\sum_{n=0}^M(b^{-1}z^p)^n
 \leq\frac{1-b^{-1}}{1-b^{-1}z^p}.
\]
The diagonal restriction of the thermal semigroup is the
birth--death semigroup with upward rate $\gamma(n+1)/(b-1)$
and downward rate $\gamma bn/(b-1)$.
Its generating function gives
\[
 \langle n|\Phi_t(X_{z,M})|n\rangle
 \uparrow d_t(z)^{-1}u_t(z)^n,
\]
where
\[
 d_t(z)=1-\frac{(1-\eta)(z-1)}{b-1}>0,
 \qquad
 u_t(z)=1+\frac{\eta(z-1)}{d_t(z)}.
\]
For $t>0$, the generating function has a power series with
nonnegative coefficients and radius $(b-\eta)/(1-\eta)>b$;
thus the identity is valid also for the present $1<z<b$.
At $t=0$ it is immediate.
This uses only the transition probabilities and bounded inputs,
without applying $\Phi_t$ to an unbounded observable.

By \eqref{eq:f-beta},
\[
 \eta>\frac{f_\beta(p)}{f_\beta(q)}
 =\frac{(b^{1/q}-1)(b-b^{1/p})}
        {(b^{1/p}-1)(b-b^{1/q})},
\]
which is equivalent to $u_t(b^{1/p})>b^{1/q}$.
By continuity, choose $1<z<b^{1/p}$ such that
$u_t(z)>b^{1/q}$.
The output operators are diagonal, so monotone convergence gives
\[
 \norm{\Phi_t(X_{z,M})}_{q,\rho_\beta}^q
 \longrightarrow
 (1-b^{-1})d_t(z)^{-q}
 \sum_{n=0}^\infty\bigl(b^{-1}u_t(z)^q\bigr)^n
 =+\infty.
\]
The input norms are uniformly bounded, proving that the
$p\to q$ norm is infinite. The case $p=q$ follows directly
from weighted contractivity and the fixed identity.
\end{proof}

\begin{remark}
For example,
\begin{equation}\label{eq:f34}
 f_\beta(3)=2\cosh(\beta/6),\qquad
 f_\beta(4)=1+2\cosh(\beta/4).
\end{equation}
For large $q$,
\begin{equation}\label{eq:f-asymp}
 f_\beta(q)\sim\frac{2\sinh(\beta/2)}\beta\,q.
\end{equation}
This is useful when comparing large-exponent times with anisotropic one-mode results.
\end{remark}

Let
\begin{equation}\label{eq:etaK}
 \thetaK(t)=-2\log\|e^{t\bfZ_{\KMS}}\|\geq0, \quad t\geq 0.
\end{equation}
\begin{samepage}
\begin{proposition}\label{thm:2q}
Suppose that $\{\Phi_t\}_{t\geq 0}$ is a Gaussian quantum Markov semigroup.
Suppose that $q\geq2$.
Then 
\begin{equation}\label{eq:etaK2q}
 \norm{\Phi_t}_{2\to q,\rho}\leq 1,
 \quad \norm{\Phi_t^\sharp}_{2\to q,\rho}\leq 1,
\end{equation}
whenever $\displaystyle \thetaK(t)\geq \frac{1}{\hc}\log(q-1)$.
In particular, if $\gamma_{\KMS}>0$, then both inequalities hold whenever
\begin{equation}\label{eq:2q-rate}
 t\geq\frac1{\gamma_{\KMS}\hc}\log(q-1).
\end{equation}
\end{proposition}
\end{samepage}
\begin{proof}
By Proposition \ref{prop:second} and Equation \eqref{eq:second}, we have that 
\begin{align*}
 \Psi_s^{(2)}\simeq\cF_c(e^{-s/2}I).
\end{align*}
Put $s=\thetaK(t)$. Then we obtain that $\norm{e^{s/2}e^{t\bfZ_{\KMS}}}=1$.
Let 
\begin{align*}
    \mathcal R_t^{(2)}=\mathcal U^{-1}\cF_c(e^{s/2}e^{t\bfZ_{\KMS}})\mathcal U.
\end{align*}
We have that $\norm{\mathcal R_t^{(2)}}_{2\to2}\leq1$.
By Proposition \ref{prop:second}, we obtain that 
\begin{equation}\label{eq:hilbert-factor}
 \Phi_t^{(2)}=\Psi_s^{(2)}\mathcal R_t^{(2)}.
\end{equation}
This thermal estimate applies to every element of $L_2(\rho)$, including the possibly unbounded $L_2$ element $\mathcal R_t^{(2)}(X)$. 
By Proposition \ref{thm:thermal}, we have that 
\begin{align*}
 \np{\Phi_t(X)}q\leq\np{\mathcal R_t^{(2)}X}2\leq\np X2,
\end{align*}
whenever $\displaystyle \theta_{\KMS}(t)\geq \frac{1}{\hc}\log (q-1)$.
The same argument applies to the adjoint because $\norm{e^{t\bfZ_{\KMS}^*}}=\norm{e^{t\bfZ_{\KMS}}}$.
By Proposition \ref{prop:CK} and Equations \eqref{eq:gammaK}, \eqref{eq:oneparticle-bound}, we obtain that $s\geq\gamma_{\KMS} t$.
\end{proof}

\begin{remark}\label{rem:nonpositive}
The remainder in \eqref{eq:hilbert-factor} is only asserted to be an $L_2$ contraction. It need not be a positive map. 
It is not applied to positive inputs in a reverse inequality. The reverse argument in \Cref{sec:reverse} uses the actual completely positive map $\Phi_t^\sharp$ instead.
\end{remark}

For $1<p\leq q<\infty$, let $p'=p/(p-1)$ be the conjugate exponent of $p$, and define
\begin{equation}\label{eq:Xi}
 \Xi(p,q)=
 \begin{cases}
 \displaystyle \log\left(\frac{2q}{p}-1\right),&2\leq p\leq q,\\[2pt]
 \displaystyle \log\left(\frac{2p'}{q'}-1\right),&p\leq q\leq2,\\[2pt]
 \displaystyle \log\left(\frac{q-1}{p-1}\right),&p\leq2\leq q.
 \end{cases}
\end{equation}
At the common boundaries these definitions agree.
Now we are able to state the hypercontractivity for the Gaussian quantum Markov semigroups.
\begin{samepage}
\begin{theorem}\label{thm:forward}
Suppose that $\gamma_{\KMS}>0$. 
For $1<p\leq q<\infty$,
\begin{equation}\label{eq:forward-piecewise}
 \norm{\Phi_t}_{p\to q,\rho}\leq1,
 \quad\norm{\Phi_t^\sharp}_{p\to q,\rho}\leq 1,
\end{equation}
whenever $\displaystyle t\geq\frac{\Xi(p,q)}{\gamma_{\KMS}\hc}$.
In particular, for all $X\in \cB(\cF_s(\bC^d))$,
\begin{align*}
 \np{\Phi_t(X)}q&\leq\np Xp
 ,\label{eq:preview-forward}
 \end{align*}
 whenever $\displaystyle t\geq\frac{2}{\gamma_{\KMS}\hc}\log\frac{q-1}{p-1}$.
\end{theorem}
\end{samepage}
\begin{proof}
Let $a=\gamma_{\KMS}\hc$.
If $2\leq p\leq q$, interpolate the $2\to r$ bound with $\infty\to\infty$ contractivity, using interpolation parameter $2/p$. The output exponent is $rp/2$, so choose $r=2q/p$. The time in \eqref{eq:2q-rate} is $a^{-1}\log(2q/p-1)$.

If $p\leq q \leq2$, use weighted duality \eqref{eq:forward-dual} and the preceding argument for $\Phi_t^\sharp:L_{q'}(\rho)\to L_{p'}(\rho)$. 

If $p\leq 2\leq q$, first apply the dual $p\to2$ estimate for time $a^{-1}\log(p'-1)$, and then the $2\to q$ estimate for time $a^{-1}\log(q-1)$. 
The semigroup property gives their sum, namely $\displaystyle a^{-1}\log\left(\frac{q-1}{p-1}\right)$. 
Interchanging $\Phi$ and $\Phi^\sharp$ proves the same statements for the KMS-adjoint.

To obtain the uniform bound, write $\displaystyle r=\frac{q-1}{p-1}\geq 1$. 
In the first case,
\[
 \frac{2q}{p}-1=1+\frac{2(p-1)}p(r-1)\leq1+2(r-1)\leq r^2.
\]
The second case is the dual version and the third already has coefficient one. 
Hence $\Xi(p,q)\leq2\log r$.
This completes the proof of the theorem.
\end{proof}

\begin{remark}
    No symmetry assumption is needed in Theorem \ref{thm:forward}.
\end{remark}

\begin{corollary}\label{cor:KMS-HC}
Let $\{\Phi_t\}_{t\geq 0}$ be a KMS-symmetric Gaussian quantum Markov semigroup with Hurwitz drift.
Then $\{\Phi_t\}_{t\geq 0}$ is hypercontractive.
\end{corollary}
\begin{proof}
It follows from Proposition \ref{prop:hurwitzpos} and Theorem \ref{thm:forward}.
\end{proof}

We shall discuss more general cases in the following.
Suppose only that $\bfZ$ is Hurwitz without assumption $\gamma_{\KMS}>0$. 
The function $\thetaK$ is nondecreasing and tends to infinity. We define
\begin{equation}\label{eq:tau}
 \varphi(r)=\inf\{t\geq0:\thetaK(t)\geq r\},\qquad \varphi(0)=0.
\end{equation}
It is finite. 
Now we define the following sufficient time:
\begin{equation}\label{eq:timeprofile}
 \mathfrak t(p,q)=
 \begin{cases}
 \varphi\bigl(\hc^{-1}\log(2q/p-1)\bigr),&2\leq p\leq q,\\
 \varphi\bigl(\hc^{-1}\log(2p'/q'-1)\bigr),&p\leq q\leq2,\\
 \varphi\bigl(\hc^{-1}\log(p'-1)\bigr)
 +\varphi\bigl(\hc^{-1}\log(q-1)\bigr),&p\leq2\leq q.
 \end{cases}
\end{equation}
\begin{corollary}\label{cor:eventual}
Suppose that $\bfZ$ is Hurwitz.
Then, for every $1<p\leq q<\infty$,
\[
 \norm{\Phi_t}_{p\to q,\rho}\leq1,
 \quad\norm{\Phi_t^\sharp}_{p\to q,\rho}\leq1,
\]
whenever $t\geq\mathfrak t(p,q)$.
Thus every fixed forward exponent pair is eventually attained.
\end{corollary}
\begin{proof}
The argument in Theorem \ref{thm:forward} applies using Proposition~\ref{thm:2q} and the hitting times \eqref{eq:tau}.
Since $\bfZ_{\KMS}$ is similar to a Hurwitz matrix, $e^{t\bfZ_{\KMS}} \to 0$, proving the asserted finiteness. 
Monotonicity follows from $\norm{e^{s\bfZ_{\KMS}}}\leq1$ and the semigroup property.
\end{proof}
The two times in the third line of \eqref{eq:timeprofile} must be added. 
A nonlinear clock cannot in general be substituted for a linear rate without accounting for this composition.

\section{Reverse Hypercontractivity}\label{sec:reverse}

In this section, we shall obtain the reverse hypercontractivity for Gaussian quantum Markov semigroups.

\begin{lemma}\label{lem:transfer}
Suppose that $\Phi$ is a normal unital completely positive map on $\cB(\cH)$, with a faithful invariant normal state $\rho$.
Fix $\displaystyle \frac12<p<q<1$.
If $\norm{\Phi^\sharp}_{\frac{p}{1-p}\to \frac{q}{1-q},\rho}\leq1$, then for $0\leq X\in \cB(\cH)$,
\begin{equation}\label{eq:transfer}
 \np{\Phi(X)}p\geq\np Xq.
\end{equation}
\end{lemma}
\begin{proof}
First let $Y\in\cB(\cH)_{++}$, meaning that $Y$ is bounded and bounded below by a strictly positive multiple of $1$.
Since $\Phi^\sharp$ is unital, $\Phi^\sharp(Y)$ is also bounded and uniformly positive.
Operator Jensen's inequality for the inverse implies
\[
 (\Phi^\sharp (Y))^{-1}\leq\Phi^\sharp(Y^{-1}).
\]
Note that
\begin{align}
 \np{\Phi^\sharp (Y)}{-\frac{q}{1-q}}
 &=\np{(\Phi^\sharp Y)^{-1}}{\frac{q}{1-q}}^{-1}\notag\\
 &\geq\np{\Phi^\sharp(Y^{-1})}{\frac{q}{1-q}}^{-1}
 \geq\np{Y^{-1}}{\frac{p}{1-p}}^{-1}=\np Y{-{\frac{p}{1-p}}}.\label{eq:negtransfer}
\end{align}
For $r\in(0,1)$, the weighted reverse H\"older duality is
\begin{equation}\label{eq:reverse-dual}
 \np Xr=\inf_{Y\in\cB(\cH)_{++}}
 \frac{\langle X,Y\rangle_\rho}{\np Y{r/(r-1)}}.
\end{equation}
The infimum is taken over the same bounded uniformly positive operators.
The reverse H\"older inequality gives the lower bound in \eqref{eq:reverse-dual}.
To verify equality in this test class, first suppose that $\varepsilon I\leq X\leq MI$, with $\varepsilon>0$, and put
\[
 a=\frac{r}{1-r},\qquad A=\rho^{1/(2r)}X\rho^{1/(2r)},\qquad
 Y_r=\rho^{1/(2a)}A^{r-1}\rho^{1/(2a)}.
\]
The last expression is understood as a quadratic form.
Since $r-1\in(-1,0)$, operator antitonicity of $u\mapsto u^{r-1}$ gives
\[
 M^{r-1}I\leq Y_r\leq\varepsilon^{r-1}I,
\]
so $Y_r\in\cB(\cH)_{++}$.
For completeness, let $B=\rho^{1/(2a)}$ and $C=A^{(1-r)/2}$.
The inequalities $\varepsilon^{1-r}B^2\leq C^2\leq M^{1-r}B^2$ show that $T=C^{-1}B$ extends to a bounded invertible operator.
Thus $Y_r=T^*T$ and $\rho^{1/(2a)}Y_r^{-1}\rho^{1/(2a)}=A^{1-r}$.
Consequently,
\[
 \langle X,Y_r\rangle_\rho=\Tr(A^r),\qquad
 \np{Y_r}{-a}=\bigl(\Tr(A^r)\bigr)^{-1/a},
\]
which proves equality.
For a general bounded $X\geq0$, use the minimizing operator for $X+\varepsilon I$.
Its quotient with $X$ is at most $\np{X+\varepsilon I}r$, which tends to $\np Xr$ as $\varepsilon\downarrow0$.
This proves \eqref{eq:reverse-dual} for every bounded positive $X$.

By Equation \eqref{eq:reverse-dual} and the reverse H\"older duality, we see that 
\begin{align*}
 \np{\Phi (X)}p
 &=\inf_{Y\in\cB(\cH)_{++}}\frac{\langle X,\Phi^\sharp (Y)\rangle_\rho}{\np Y{-{\frac{p}{1-p}}}}\\
 &\geq\np Xq\inf_{Y\in\cB(\cH)_{++}}
 \frac{\np{\Phi^\sharp (Y)}{-\frac{q}{1-q}}}{\np Y{-{\frac{p}{1-p}}}}
 \geq\np Xq.
\end{align*}
This completes the proof of the lemma.
\end{proof}


\begin{theorem}\label{thm:reverse}
Suppose that $\gamma_{\KMS}>0$ and $1/2<p<q<1$.
We have that for $X\geq 0$,
\begin{equation}\label{eq:reverse-piecewise}
 \np{\Phi_t(X)}p\geq\np Xq\quad
\end{equation}
whenever $\displaystyle t\geq\frac{\Xi(p/(1-p),q/(1-q))}{\gamma_{\KMS}\hc}$.
In particular, the uniform time 
\begin{align}
 t\geq\frac{2}{\gamma_{\KMS}\hc}
 \log\frac{(2q-1)(1-p)}{(2p-1)(1-q)}.\label{eq:preview-reverse}
\end{align}
is valid. 
If $p\leq2/3\leq q$, the factor two in that uniform time \eqref{eq:preview-reverse} can be omitted.

If $\bfZ$ is Hurwitz without assumption that $\gamma_{\KMS}>0$, the same reverse inequality holds for $t\geq\mathfrak t(p/(1-p),q/(1-q))$, with \eqref{eq:timeprofile}. In particular, every fixed pair in $(1/2,1)$ is eventually attained.
\end{theorem}
\begin{proof}
By Lemma~\ref{lem:transfer} for $\Phi_t$, Theorem~\ref{thm:forward}, and Corollary~\ref{cor:eventual} for $\Phi_t^\sharp$, we see that the theorem is true.
The transformed exponents straddle two precisely when $p\leq2/3\leq q$.
\end{proof}

\begin{corollary}\label{cor:KMS-reverse}
Let $\{\Phi_t\}_{t\geq 0}$ be a KMS-symmetric Gaussian quantum Markov semigroup with Hurwitz drift.
Then $\{\Phi_t\}_{t\geq 0}$ has reverse hypercontractivity throughout $1/2<p<q<1$, with a curve starting at time zero.
\end{corollary}
\begin{proof}
It follows from Theorem \ref{thm:reverse}.    
\end{proof}

\begin{remark}
 The quantum diffusion matrix of the KMS-symmetric quantum Markov semigroup may be singular in Corollary \ref{cor:KMS-reverse}.
\end{remark}

Write $\displaystyle c=\frac{\gamma_{\KMS}\hc}{2}$. 
Given $q\in (1/2,1)$, let
\begin{equation}\label{eq:reverse-curve}
 z_q=\frac{2q-1}{1-q},\qquad
 p(t)=\frac{1+z_qe^{-ct}}{2+z_qe^{-ct}}.
\end{equation}
Then $p(0)=q$, $p(t)>1/2$, and
\begin{equation}\label{eq:curve-ode}
 p'(t)=-c(2p(t)-1)(1-p(t)),\qquad
 \np{\Phi_t(X)}{p(t)}\geq\np Xq.
\end{equation}
There is no positive waiting time before the exponent starts to improve. Reaching the boundary $1/2$ is not asserted at a finite time.

The logarithmic Sobolev quantities and their form domains are defined in Section \ref{sec:lsi}. 
Differentiating this curve there gives
\[
 \mathcal E_r(X)\geq\frac c4(2r-1)\Ent_{r,\rho}(X),\qquad \frac12<r<1.
\]
The reverse proof itself does not depend on differentiating an unbounded generator. The boundary is structural for this method: $r/(1-r)>1$ exactly when $r>1/2$. 
The primitive example in Section~\ref{sec:lsi} shows that KMS symmetry and stability alone cannot supply reverse curves with a strictly negative initial exponent derivative throughout $(0,1)$.
This is an obstruction at time zero; it does not rule out eventual reverse inequalities for fixed pairs $0<p<q<1/2$.

\section{Examples for (Reverse) Hypercontractivity}

In this section, we present some examples of Gaussian quantum Markov semigroups with hypercontractivity and reverse hypercontractivity.
We also provide examples of Gaussian quantum Markov semigroups that fail hypercontractivity and reverse hypercontractivity.

\subsection{A KMS-symmetric family}\label{subsec:rankone}
We specialize Example~\ref{ex:anisotropic-boundary} to
\begin{equation}\label{eq:rank-ZC}
 \nu=\cosh r,\qquad r_1=\lambda e^r,\qquad r_2=\lambda e^{-r},
 \qquad \lambda>0,\quad r>0.
\end{equation}
The complete-positivity condition \eqref{eq:anisotropic-cp} is an equality.
Thus \eqref{eq:anisotropic-Q} has rank one, and
\eqref{eq:anisotropic-rate} gives $\alpha_{\mathrm{CMLSI}}(\Phi;\mathbb E_\Phi)=0$.
The same example shows that the semigroup is KMS-symmetric but not
GNS-symmetric, and that
\begin{equation}\label{eq:rank-rate}
 \gamma_{\KMS}=2\lambda e^{-r}>0.
\end{equation}
Its faithful invariant density is $\rho_\beta$, with
\begin{equation}\label{eq:rank-beta}
 \beta=2\log\coth(r/2).
\end{equation}

A concrete GKSL realization uses one jump operator and a squeezing Hamiltonian:
\begin{align}
 L_1&=\sqrt{2\lambda\cosh r}
\left(\cosh(r/2)\bfa+\sinh(r/2)\bfa^\dagger\right),\label{eq:rank-L}\\
 H&=\frac{i\lambda\sinh r}{2}\bigl((\bfa^\dagger)^2-\bfa^2\bigr).
 \label{eq:rank-H}
\end{align}
With the convention
$-\cL(X)=i[H,X]-\frac12\{L_1^*L_1,X\}+L_1^*XL_1$,
substitution into \eqref{eq:ZUV}--\eqref{eq:CUV} gives precisely
\eqref{eq:anisotropic-example} with the parameters \eqref{eq:rank-ZC}.
The positive KMS rate therefore supplies the following norm estimates
even though the CMLSI constant vanishes.

By Propositions~\ref{thm:2q} and~\ref{prop:thermal-profile}, we have
\begin{equation}\label{eq:rank-forward}
 \norm{\Phi_t}_{2\to q,\rho_\beta}\leq1,
\end{equation}
whenever $\displaystyle t\geq\frac1{2\lambda e^{-r}}\log f_\beta(q)$.

For all $1/2<p<q<1$, Theorem \ref{thm:reverse} gives the explicit piecewise time using \eqref{eq:rank-rate}. 
For an explicit pair $p=2/3$, $q=3/4$, KMS symmetry and Lemma \ref{lem:transfer} give
\begin{equation}\label{eq:rank-reverse}
 \np{\Phi_t(X)}{2/3}\geq\np X{3/4}.
\end{equation}
whenever $\displaystyle t\geq\frac1{2\lambda e^{-r}}\log\bigl(2\cosh(\beta/6)\bigr)$.

Set $\lambda=2$, $r=\log2$. 
Then
\begin{gather}
 \bfZ=\diag(-4,-1),\quad \Sigma=\tfrac54I_2,
 \quad \bfC=\diag(10,5/2),\quad \beta=\log9,\label{eq:num-data}\\
 \gamma_{\KMS}=2,\qquad \alpha_{\mathrm{CMLSI}}(\Phi;\mathbb E_\Phi)=0.
 \label{eq:num-rates}
\end{gather}

\subsection{Two modes coupled to a single bath}\label{subsec:twomode}\label{subsec:norate}
Fix $\gamma>0$, $\beta>0$, and $g\in\mathbb R\setminus\{0\}$.
Only the first mode is coupled to the bath, while the Hamiltonian
$H=g(\bfa_1^\dagger\bfa_2+\bfa_2^\dagger\bfa_1)$ couples the two modes.
The generator is
\begin{equation}\label{eq:indirectgen}
\begin{aligned}
 -\cL(X)={}&i[H,X]\\
 &+\frac{\gamma e^\beta}{e^\beta-1}
 \left(\bfa_1^\dag X\bfa_1-\frac12\{\bfa_1^\dag\bfa_1,X\}\right)\\
 &+\frac{\gamma}{e^\beta-1}
 \left(\bfa_1X\bfa_1^\dag-\frac12\{\bfa_1\bfa_1^\dag,X\}\right).
\end{aligned}
\end{equation}
Put $\nu=\coth(\beta/2)$, $D=\diag(\gamma,0)$, and
$\Omega=\left(\begin{smallmatrix}0&g\\g&0\end{smallmatrix}\right)$.
The common invariant state is $\rho=\rho_\beta^{\otimes2}$, with
\begin{equation}\label{eq:indirectmatrices}
 \bfZ=\begin{pmatrix}-D/2&-\Omega\\\Omega&-D/2\end{pmatrix},
 \qquad \bfC=\nu\diag(D,D),\qquad\Sigma=\nu I_4.
\end{equation}
Equivalently, the complex Weyl-label drift is
\begin{equation}\label{eq:K2}
 K=-D/2+i\Omega=\begin{pmatrix}-\gamma/2&ig\\ig&0\end{pmatrix}.
\end{equation}
The annihilation-operator drift has the opposite Hamiltonian sign
and the same singular values.
The eigenvalues of $K$ satisfy $\lambda^2+(\gamma/2)\lambda+g^2=0$,
so $K$ is Hurwitz for every $g\ne0$.
By Proposition~\ref{prop:control}, the semigroup is irreducible and its faithful
invariant state is unique.
It is not KMS-symmetric, since the Hamiltonian contributes a nonzero
skew-adjoint part. Moreover,
\[
 \Sigma_{\KMS}=\sqrt{\nu^2-1}\,I_4,\qquad
 \bfC_{\KMS}=\sqrt{\nu^2-1}\,\diag(D,D),\qquad
 \gamma_{\KMS}=0.
\]
Both the quantum-diffusion and KMS-dissipation matrices have a kernel
on the undamped mode, although the drift is Hurwitz.

For the whole family, the time function is
$\thetaK(t)=-2\log\|e^{tK}\|_{\mathrm{op}}$.
It is nondecreasing because $K+K^*=-D\leq0$, and tends to infinity
by Hurwitz stability.
The common temperature and complex-linear drift give the following
factorization for every $g\ne0$.
Here $\mathbf F=\nu I_4+i\mathbf J$.
Fix $t\geq0$ and put $s=\thetaK(t)$ and $\mathbf S=e^{s/2}e^{t\bfZ}$.
Since $\Sigma_{\KMS}$ is scalar, the definition of $s$ gives $\|\mathbf S\|_{\mathrm{op}}=1$.
The real matrix $\bfZ$ of $K$ commutes with $\mathbf J$, so $\mathbf S\mathbf J=\mathbf J\mathbf S$. Hence
\[
 \mathbf F-\mathbf S^{\mathsf T}\mathbf F\mathbf S
 =\mathbf F(I_4-\mathbf S^{\mathsf T}\mathbf S)\geq0,
\]
because the two factors on the right are commuting positive semidefinite matrices.
The Gaussian complete-positivity criterion \cite{CEGH08} therefore gives a normal $\rho$-preserving unital completely positive map
\[
 \mathcal R(W(\mathbf z))
 =\exp\!\left(-\frac12\mathbf z^{\mathsf T}
 (\Sigma-\mathbf S^{\mathsf T}\Sigma\mathbf S)\mathbf z\right)W(\mathbf S\mathbf z).
\]
Composition with \eqref{eq:thermal-ref} yields
\[
 \Phi_t=\mathcal R\circ\Psi_s=\Psi_s\circ\mathcal R.
\]
Since $\mathcal R$ is contractive on each weighted $L_p(\rho)$, $1\leq p\leq\infty$, Proposition~\ref{thm:thermal} transfers the full thermal hypercontractivity curve to $\Phi_t$.
Taking KMS adjoints gives the same forward bound for $\Phi_t^\sharp$, since $\Psi_s^\sharp=\Psi_s$ and $\mathcal R^\sharp$ is also unital, completely positive, and $\rho$-preserving.
Thus forward hypercontractivity holds for $1<p\leq q<\infty$ when
\begin{equation}\label{eq:2mode-forward}
 \thetaK(t)\geq\hc^{-1}\log\frac{q-1}{p-1}
\end{equation}
and Lemma~\ref{lem:transfer} gives reverse hypercontractivity on $1/2<p<q<1$ when
\begin{equation}\label{eq:2mode-reverse}
 \thetaK(t)\geq\hc^{-1}
 \log\frac{(2q-1)(1-p)}{(2p-1)(1-q)}.
\end{equation}
Thus every exponent pair in these ranges has a finite sufficient
time. The next proposition rules out a positive initial slope for
forward exponent curves and a negative initial slope for reverse curves.

\begin{proposition}\label{prop:norate}
For the Gaussian quantum Markov semigroup with generator \eqref{eq:indirectgen}, no uniform forward exponent curve $q(t)$ with $q(0)=r>1$ and $q'(0)>0$ can satisfy
\begin{align*}
\np{\Phi_t(X)}{q(t)}\leq\np Xr
\end{align*}
for all inputs near $t=0$. 

No uniform reverse exponent curve $p(t)$ with $p(0)=r\in(1/2,1)$ and $p'(0)<0$ can satisfy 
\begin{align*}
\np{\Phi_t(X)}{p(t)}\geq\np Xr
\end{align*}
for all positive inputs near $t=0$.
\end{proposition}
\begin{proof}
Fix $\varepsilon>0$, put $a=e^{-\beta}$, and let
\[
 P_M=\sum_{n=0}^M|n\rangle\langle n|,\qquad
 Q_0=|0\rangle\langle0|,\qquad
 X_M=I+\varepsilon P_M\otimes Q_0.
\]
The operator $X_M-I$ has finite occupation-number support and is a Schwartz operator, so $X_M\in\operatorname{Dom}(\cL)$; see \cite[Section~3.1]{FG26}.
Moreover, $X_M\geq I$ and $\cL(X_M)$ has finite matrix support.
Thus the weighted trace differentiations below are justified.

Since $[H,\rho]=0$, the Hamiltonian contributes zero to the fixed-exponent derivative.
For the bath dissipator, only the first-mode edge $M\leftrightarrow M+1$ contributes.
The thermal upward and downward rates at occupation $n$ are $\gamma a(n+1)/(1-a)$ and $\gamma n/(1-a)$, respectively.
Using the equilibrium weights $(1-a)a^n$ and the second-mode vacuum probability $1-a$, we obtain
\[
 \Tr\bigl(\rho X_M^{r-1}\cL(X_M)\bigr)
 =(1-a)\gamma(M+1)a^{M+1}\varepsilon
 \bigl((1+\varepsilon)^{r-1}-1\bigr)\longrightarrow0.
\]
Since $X_M$ commutes with $\rho$, it follows that
\[
 \left.\frac{d}{dt}\log\np{\Phi_t(X_M)}r\right|_{t=0}
 =-\frac{\Tr\bigl(\rho X_M^{r-1}\cL(X_M)\bigr)}{\np{X_M}r^r}
 \longrightarrow0.
\]
On the other hand, with $s_M=(1-a)(1-a^{M+1})$,
\[
 \np{X_M}r^r=1-s_M+s_M(1+\varepsilon)^r,
 \qquad s_M\longrightarrow1-a\in(0,1).
\]
The exponent derivative $\partial_r\log\np{X_M}r$ therefore tends to a strictly positive number.
Indeed, it is the relative entropy of the probabilities $(1-s_M,s_M)$ tilted by the $r$th powers of $(1,1+\varepsilon)$ with respect to the original probabilities, divided by $r^2$.

For any proposed forward curve with $q'(0)>0$, the chain rule gives a strictly positive total initial derivative when $M$ is sufficiently large, contradicting the forward inequality.
For any proposed reverse curve with $p'(0)<0$, it gives a strictly negative total initial derivative when $M$ is sufficiently large, contradicting the reverse inequality.
This completes the proof of the proposition.
\end{proof}

\subsection{Critical coupling and explicit time bounds}\label{subsec:criticalcoupling}
Specialize the model in Section~\ref{subsec:twomode} to $g=\gamma/4$.
The two eigenvalues of \eqref{eq:K2} coincide at $-\gamma/4$.
For $u=\gamma t/4$,
\begin{equation}\label{eq:expK2}
 e^{tK}=e^{-u}\begin{pmatrix}1-u&iu\\iu&1+u\end{pmatrix}.
\end{equation}
The two singular values of the matrix without $e^{-u}$ are
$\sqrt{1+u^2}\pm u$. Hence
\begin{equation}\label{eq:eta2}
 \thetaK(t)=2\bigl(u-\arsinh u\bigr),\qquad u=\gamma t/4.
\end{equation}
Substitution into \eqref{eq:2mode-forward}--\eqref{eq:2mode-reverse}
gives explicit sufficient times for forward and reverse hypercontractivity.
This function is strictly increasing for $t>0$ and tends to infinity,
while at zero
\begin{equation}\label{eq:cubic}
 \thetaK(t)=\frac{\gamma^3}{192}t^3+O(t^5).
\end{equation}
Its cubic behavior is consistent with the obstruction in
Proposition~\ref{prop:norate}.

For $\gamma=4$ and $\beta=\log9$, the Hamiltonian coupling and
the two dissipative coefficients in \eqref{eq:indirectgen} become
\begin{equation}\label{eq:network-numeric}
 g=1,\qquad
 \frac{\gamma e^\beta}{e^\beta-1}=\frac92,\qquad
 \frac{\gamma}{e^\beta-1}=\frac12.
\end{equation}
The equal-temperature value \eqref{eq:h-equal} is
\[
 \hc=\frac2{2+\log45},\qquad
 \thetaK(4)=2(4-\arsinh4)\simeq3.810575.
\]
Since $\hc^{-1}\log3\simeq3.189635<\thetaK(4)$,
\eqref{eq:2mode-forward}--\eqref{eq:2mode-reverse} give
\begin{equation}\label{eq:network-both}
 \np{\Phi_4(X)}4\leq\np X2,\qquad
 \np{\Phi_4(X)}{2/3}\geq\np X{4/5},
\end{equation}
where $X\geq0$ in the second inequality.

\subsection{Countably many modes: normality and the role of uniformity}\label{subsec:infiniteconstruction}
In this section, we focus on the infinitely many modes.
This shall explain why ``finitely many modes'' cannot be dropped without further uniform estimates.

Let $\cH_\infty=\cF_s(\ell^2(\mathbb N))$, $\beta_j=j$, and
\begin{equation}\label{eq:infiniterho}
 \rho_\infty=\mathcal Z^{-1}\exp\!\left(-\sum_{j\geq1}j\bfa_j^\dag\bfa_j\right),
 \qquad \mathcal Z=\prod_{j\geq1}(1-e^{-j})^{-1}<\infty.
\end{equation}
It is faithful and normal on the physical Fock space: 
finite occupation vectors form a basis, every eigenvalue is positive, and the product partition function converges. 
Its mean total particle number is finite because $\displaystyle \sum_j(e^j-1)^{-1}<\infty$.

For positive bounded rates $\gamma_j$, put $K_t e_j=e^{-\gamma_jt/2}e_j$ and $\Sigma_\infty e_j=\coth(j/2)e_j$.
We construct a normal quantum Markov semigroup with Weyl action
\begin{equation}\label{eq:infiniteweyl}
 \Phi_t(W(z))=
 \exp\!\left(-\frac12\langle z,(\Sigma_\infty-K_t^*\Sigma_\infty K_t)z\rangle\right)
 W(K_tz).
\end{equation}
On two copies of $\ell^2(\mathbb N)$, the block operator
\[
 \begin{pmatrix}
 K_t&(I-K_t^2)^{1/2}\\
 -(I-K_t^2)^{1/2}&K_t
 \end{pmatrix}
\]
is unitary. Let $U_t$ be its second quantization, using the identification
$\cF_s(\ell^2\oplus\ell^2)=\cH_\infty\otimes\cH_\infty$.
With $\omega_\infty(Y)=\Tr(\rho_\infty Y)$, define
\[
 \Phi_t(X)=(\operatorname{id}\otimes\omega_\infty)
 \bigl(U_t^*(X\otimes I)U_t\bigr).
\]
Since $\omega_\infty$ is normal, this is a normal unital completely positive map.
The second-quantization identity gives
\[
 U_t^*(W(z)\otimes I)U_t
 =W(K_tz)\otimes W((I-K_t^2)^{1/2}z).
\]
Taking the thermal expectation proves \eqref{eq:infiniteweyl}, because $K_t$ commutes with $\Sigma_\infty$.
The identity $K_{s+t}=K_sK_t$ and the Gaussian factors in \eqref{eq:infiniteweyl} give the semigroup law on Weyl operators; normality and the ultraweak density of their linear span extend it to $\cB(\cH_\infty)$.
Likewise, the thermal characteristic function gives $\omega_\infty\circ\Phi_t=\omega_\infty$.
The one-particle unitaries, and hence $U_t$, converge strongly to the identity as $t\downarrow0$.
For every trace-class $T$, the predual therefore satisfies
\[
 \|\Phi_{t,*}(T)-T\|_1
 \leq\|U_t(T\otimes\rho_\infty)U_t^*-T\otimes\rho_\infty\|_1
 \longrightarrow0.
\]
Thus the predual semigroup is strongly continuous in trace norm, and $\Phi_t$ is weak$^*$ continuous.

The weak$^*$ generator is given on finite-mode, finite-particle matrix units by the GKSL expression
\begin{equation}\label{eq:infinitegen}
 -\cL(X)=\sum_{j\geq1}\gamma_j
 \left(\frac{e^j}{e^{j}-1}\left(\bfa_j^\dag X\bfa_j-\frac{1}{2}\{\bfa_j^\dag \bfa_j, X\}\right)+\frac{1}{e^{j}-1}\left(\bfa_j X\bfa_j^\dag-\frac{1}{2}\{\bfa_j \bfa_j^\dag, X\}\right)\right).
\end{equation}
More precisely, let $F\subset\mathbb N$ be finite and let $X=X_F\otimes I$, where $X_F$ is a finite linear combination of occupation-basis matrix units on $\cF_s(\ell^2(F))$.
On $\cB(\cF_s(\ell^2(F)))\otimes I$, the channel restricts to the finite-mode thermal semigroup, so such $X$ belongs to the weak$^*$ generator domain and \eqref{eq:infinitegen} follows by finite-mode differentiation.
For $j\notin F$ the corresponding summand vanishes on finite-particle vectors and extends as zero; for $j\in F$ all products in that summand extend to bounded operators.
Thus \eqref{eq:infinitegen} is a finite sum on this test class.
The semigroup is GNS symmetric.

If every $\gamma_j>0$, then $K_tz\to0$ for every $z\in\ell^2$ by dominated convergence. The characteristic-function invariance argument used in Lemma~\ref{lem:endstructure} shows that \eqref{eq:infiniterho} is the unique invariant normal state. In $L_2(\rho_\infty)$, finite products of centered thermal-chaos vectors, together with $1$, are total and every nonconstant product has strictly positive decay rate. Hence the fixed algebra is scalar. These examples are primitive in this explicit sense, though their convergence need not be uniform over all inputs.

Set $\gamma_j=1$ for every mode. The product thermal-chaos basis shows that the KMS numerical-range gap is $1/2$: the $n$th total chaos decays by $e^{-nt/2}$ and the first chaos attains the rate. In particular, this example has a positive gap independent of the mode index.

\begin{proposition}\label{prop:infiniteforward}
For the primitive GNS-symmetric semigroup \eqref{eq:infiniteweyl} with $\gamma_j=1$,
\begin{equation}\label{eq:infiniteforwardfail}
 \norm{\Phi_t}_{p\to q,\rho_\infty}=\infty
 \quad(t<\infty,\ 1<p<q<\infty).
\end{equation}
\end{proposition}
\begin{proof}
Let $P_j$ be the occupation-one projection on mode $j$, tensored with identities on the other modes. Its equilibrium probability is
\[
 s_j=(1-e^{-j})e^{-j},\qquad \norm{P_j}_{p,\rho_\infty}=s_j^{1/p}.
\]
The restriction of the thermal generator to diagonal observables is the birth--death generator with upward rate $\displaystyle \frac{n+1}{e^{j}-1}$ and downward rate $\displaystyle \frac{n e^j}{e^j-1}$. 
Starting from $n=1$, the probability of no jump by time $t$ is $e^{-(1+\frac{3}{e^{j}-1})t}$. 
Therefore
\[
 (\Phi_t(P_j))(1)\geq e^{-(1+\frac{3}{e^{j}-1})t},\qquad
 \frac{\norm{\Phi_t(P_j)}_{q,\rho_\infty}}{\norm{P_j}_{p,\rho_\infty}}
 \geq e^{-(1+\frac{3}{e^{j}-1})t}s_j^{1/q-1/p}\longrightarrow\infty.
\]
This proves \eqref{eq:infiniteforwardfail} using bounded local witnesses.
\end{proof}


Now set $\gamma_j=1/j$, with the same invariant density \eqref{eq:infiniterho}. 
Each mode relaxes, but there is no uniform positive relaxation rate.

\begin{proposition}\label{prop:infinitereverse}
For every finite $t$ and every fixed $1/2<p<q<1$, there exists a bounded $X\geq0$ such that
\begin{equation}\label{eq:infinitereversefail}
 \norm{\Phi_t(X)}_{p,\rho_\infty}<\norm X_{q,\rho_\infty}.
\end{equation}
\end{proposition}
\begin{proof}
Let $X_{j,\varepsilon}=e^{-\varepsilon \bfa_j^\dag\bfa_j}$, $\varepsilon>0$ and write
\[
 v_j=\frac{e^{j}}{(e^{j}-1)^2}>0.
\]
The mean-number formula is
\[
 \Phi_t(\bfa_j^\dag\bfa_j)=\frac{1}{e^{j}-1}+e^{-\gamma_jt}(\bfa_j^\dag\bfa_j-\frac{1}{e^{j}-1}).
\]
Since all thermal number moments are finite, expansion at $\varepsilon=0$ gives
\begin{equation}\label{eq:reverseexpansion}
 \log\norm{\Phi_t(X_{j,\varepsilon})}_{p,\rho_\infty}
 -\log\norm{X_{j,\varepsilon}}_{q,\rho_\infty}
 =\frac{\varepsilon^2v_j}{2}
 \left((1-q)-(1-p)e^{-2\gamma_jt}\right)+o(\varepsilon^2).
\end{equation}
A necessary condition for the reverse inequality is therefore
\[
 t\geq\frac1{2\gamma_j}\log\frac{1-p}{1-q}\quad\text{for every }j.
\]
For $\gamma_j=1/j$, these necessary times diverge. Given $t$, choose $j$ so that the coefficient in \eqref{eq:reverseexpansion} is strictly negative, then choose $\varepsilon>0$ sufficiently small. 
The witness is bounded. 
Adding a sufficiently small multiple of $1$ also gives a bounded uniformly positive witness to the strict violation.
\end{proof}

An exact evaluation is available. Put $z=e^{-\varepsilon}$, $\theta=e^{-\gamma_jt}$, and
\[
 d_t=1+\frac{1}{e^{j}-1}(1-\theta)(1-z),\qquad c_t=d_t^{-1},
 \qquad u_t=1-\frac{\theta(1-z)}{d_t}.
\]
The birth--death generating-function equation gives $\Phi_t(z^{\bfa_j^\dag\bfa_j})=c_tu_t^{\bfa_j^\dag\bfa_j}$. 
Consequently,
\begin{align}
 \norm{\Phi_t(X_{j,\varepsilon})}_{p,\rho_\infty}
 &=c_t\left(\frac{1-e^{-j}}{1-e^{-j}u_t^p}\right)^{1/p},\label{eq:reverseexact1}\\
 \norm{X_{j,\varepsilon}}_{q,\rho_\infty}
 &=\left(\frac{1-e^{-j}}{1-e^{-j}e^{-q\varepsilon}}\right)^{1/q}.
 \label{eq:reverseexact2}
\end{align}
For $t=1$, $j=8$, $\varepsilon=0.1$, $p=2/3$ and $q=3/4$, the value in \eqref{eq:reverseexact1} is approximately $0.999967656766767$
and the value in \eqref{eq:reverseexact2} is approximately $0.999967670922434$.
The strict inequality has the wrong sign for reverse hypercontractivity. 
The same slow-mode example also fails forward boundedness by the projection argument of Proposition~\ref{prop:infiniteforward}.

\section{Logarithmic Sobolev Inequalities}\label{sec:lsi}

In this section, we shall study the $p$-logarithmic Sobolev inequality for $0<p<\infty$.
The symbols $L_\ell$ denote the jump operators.
For a bounded $X\geq\varepsilon 1$, $\varepsilon>0$, put
 \begin{equation}\label{eq:escort}
 X_p=j_p(X)=\rho^{\frac{1}{2p}} X\rho^{\frac{1}{2p}}.
 \end{equation}
For density operators $\omega,\eta$, $H(\omega\Vert\eta)$ denotes Umegaki relative entropy, with value $+\infty$ allowed. On a class for which it is finite, define
\begin{equation}\label{eq:Entfull}
 \Ent_{p,\rho}(X)= \np Xp^p H\left(\left.\frac{X_p^p}{\np Xp^p}\right\Vert\rho\right).
\end{equation}
Throughout this section, expressions involving $\mathcal L$
are evaluated on bounded uniformly positive operators
$X\in\operatorname{Dom}(\mathcal L)$, where
$\operatorname{Dom}(\mathcal L)$ denotes the domain of the
weak$^*$-generator $-\mathcal L$.
In the variational definition below, we restrict to operators
for which the displayed trace pairings are finite and the map
\[
(r,t)\longmapsto\log\|\Phi_t(X)\|_{r,\rho}
\]
is differentiable at $(p,0)$, with derivatives given by the
trace differentiation formulas below.
Differentiation in $t$ at $t=0$ is understood from the right.
For $p\ne1$, the Dirichlet form is
\begin{equation}\label{eq:Ep}
 \mathcal E_{p,\cL}(X)
 =\frac{p^2}{4(p-1)}\Tr\!\left(X_p^{p-1}j_p(\cL (X))\right).
\end{equation}
For $0<p<1$, the trace pairing is understood as
$\operatorname{Tr}\!\left(
X_p^{(p-1)/2}\,
j_p(\mathcal L(X))\,
X_p^{(p-1)/2}
\right)$,
with the product interpreted in the quadratic-form sense.
At $p=1$ the limiting weighted energy, when finite, is
\[
 \mathcal E_1(X)=\frac14\Tr\bigl(j_1(\cL( X))(\log j_1(X)-\log\rho)\bigr).
\]
These are the entropy and $p\hat p/4$ Dirichlet normalizations of \cite{BDR20}, where $\hat p$ is the conjugate dual of $p$.
In particular, $\mathcal E_2(X)=\langle X,\cL(X)\rangle_\rho.$

The $p$-logarithmic Sobolev constant is defined by
\[
\alpha_p(\mathcal L)
=
\inf_{\substack{
X\in\operatorname{Dom}(\mathcal L),\;
X\geq\varepsilon_X I,\;\varepsilon_X>0\\
0<\operatorname{Ent}_{p,\rho}(X)<\infty,\;
|\mathcal E_{p,\mathcal L}(X)|<\infty
}}
\frac{\mathcal E_{p,\mathcal L}(X)}
{\operatorname{Ent}_{p,\rho}(X)},
\]
where the infimum is restricted to the differentiability
class specified above.
The infimum may be negative and may equal $-\infty$
when $0<p<1/2$.
On this domain, the following differentiation identities hold:
\begin{align}
 \frac{d}{dp}\log\np Xp
 &=\frac{\Ent_{p,\rho}(X)}{p^2\np Xp^p},\label{eq:pderivative}\\
 \left.\frac d{dt}\log\np{\Phi_t(X)}p\right|_{t=0}
 &=\frac{4(1-p)}{p^2\np Xp^p}\mathcal E_{p, \cL}(X).\label{eq:tderivative}
\end{align}
At $p=1$, for a KMS-symmetric quantum Markov semigroup,
\begin{equation}\label{eq:E1relative}
 \mathcal E_1(X)=\frac{\np X1}4\mathcal I\left(\frac{j_1(X)}{\np X1}\right),\qquad
 \mathcal I(\omega)=\Tr\bigl(\cL^*(\omega)(\log\omega-\log\rho)\bigr).
\end{equation}
The constants $\alpha_p(\cL)$ are defined without ancillary amplification, on the stated differentiability class.
In the KMS-symmetric case, \eqref{eq:E1relative} identifies the $p=1$ quotient with $\mathcal I(\omega)/(4H(\omega\Vert\rho))$ for the corresponding densities.
This does not identify $\alpha_1(\cL)$ with $\alpha_{\mathrm{CMLSI}}/4$: the latter constant includes all matrix amplifications and uses entropy relative to the fixed-point conditional expectation.


A negative $p$-Dirichlet form for $p<1/2$ is not negative physical relative-entropy production; 
it is failure of reverse contractivity for a different weighted functional.

\begin{proposition}\label{prop:positiveLSI}
Suppose that $\gamma_{\KMS}>0$.
For $r>1$,
\begin{equation}\label{eq:forward-LSI}
 \alpha_r(\cL)\geq \frac{\gamma_{\KMS}\hc}{8}
 \min\left\{r,\frac{r}{r-1}\right\}.
\end{equation}
In particular,
\[
 \alpha_2(\cL)\geq\frac{\gamma_{\KMS}\hc}{4}.
\]
For $1/2<r<1$,
\begin{equation}\label{eq:reverse-LSI}
 \alpha_r(\cL)\geq\frac{\gamma_{\KMS}\hc}{8}
 \begin{cases}
 \displaystyle\frac{r(2r-1)}{1-r},&\displaystyle\frac12<r\leq\frac23,\\[4pt]
 r,&\displaystyle\frac23\leq r<1.
 \end{cases}
\end{equation}
\end{proposition}
\begin{proof}
For $r>1$, choose $q(0)=r$ and $\Xi(r,q(t))=\gamma_{\KMS}\hc t$ for small $t\geq0$.
Equation~\eqref{eq:Xi} gives
\[
 q'(0)=\frac{\gamma_{\KMS}\hc}{2}\,r\min\{r-1,1\}.
\]
Let $X$ be admissible in the variational definition of $\alpha_r(\cL)$.
Theorem~\ref{thm:forward} gives
$\np{\Phi_t(X)}{q(t)}\leq\np Xr$, with equality at $t=0$.
Taking logarithms and differentiating from the right at zero, \eqref{eq:pderivative} and \eqref{eq:tderivative} yield
\[
 q'(0)\Ent_{r,\rho}(X)-4(r-1)\mathcal E_{r,\cL}(X)\leq0.
\]
Dividing by $4(r-1)\Ent_{r,\rho}(X)$ and taking the infimum proves \eqref{eq:forward-LSI}.

For $1/2<r<1$, choose $p(0)=r$ and
$\Xi(p(t)/(1-p(t)),r/(1-r))=\gamma_{\KMS}\hc t$.
Differentiating \eqref{eq:Xi} at $t=0$ gives
\[
 -p'(0)=\frac{\gamma_{\KMS}\hc}{2}
 \begin{cases}
 r(2r-1),&\displaystyle\frac12<r\leq\frac23,\\[2pt]
 r(1-r),&\displaystyle\frac23\leq r<1.
 \end{cases}
\]
Theorem~\ref{thm:reverse} gives $\np{\Phi_t(X)}{p(t)}\geq\np Xr$.
Taking logarithms and differentiating from the right at zero similarly yields
\[
 p'(0)\Ent_{r,\rho}(X)+4(1-r)\mathcal E_{r,\cL}(X)\geq0.
\]
Dividing by $4(1-r)\Ent_{r,\rho}(X)$ and taking the infimum proves \eqref{eq:reverse-LSI}.
The two expressions agree at $r=2/3$.
\end{proof}

\begin{proposition}[Endpoint nonnegativity]\label{prop:half-nonnegative}
Let $\{\Phi_t\}_{t\geq0}$ be a normal unital completely positive quantum Markov semigroup on $\cB(\cH)$ with a faithful invariant normal state $\rho$.
For every bounded $X\geq0$ and $t\geq0$,
\[
 \np{\Phi_t(X)}{1/2}\geq\np X{1/2}.
\]
For bounded uniformly positive $X\in\operatorname{Dom}(\cL)$ in the differentiability class specified above at $p=1/2$, this implies
\[
 \mathcal E_{1/2,\cL}(X)\geq0,
 \qquad \alpha_{1/2}(\cL)\geq0.
\]
\end{proposition}
\begin{proof}
For $Y\in\cB(\cH)_{++}$, operator Jensen's inequality and invariance of $\rho$ under the KMS adjoint give
\begin{align*}
 \np{\Phi_t^\sharp(Y)}{-1}
 &=\bigl[\Tr\bigl(\rho(\Phi_t^\sharp(Y))^{-1}\bigr)\bigr]^{-1}\\
 &\geq\bigl[\Tr\bigl(\rho\Phi_t^\sharp(Y^{-1})\bigr)\bigr]^{-1}
 =\np Y{-1}.
\end{align*}
Applying the weighted reverse H\"older duality \eqref{eq:reverse-dual} at $r=1/2$, as in the proof of Lemma~\ref{lem:transfer}, proves the asserted inequality for every bounded $X\geq0$.
On the stated differentiability domain, the right derivative of $\log\np{\Phi_t(X)}{1/2}$ at zero is nonnegative.
Equation~\eqref{eq:tderivative} then gives $\mathcal E_{1/2,\cL}(X)\geq0$, and the variational definition of $\alpha_{1/2}(\cL)$ proves the last assertion.
\end{proof}


We work in the Williamson representation of $\rho$, and finite matrix support is understood with respect to the corresponding occupation-number basis.
Operators with finite matrix support are Schwartz operators and belong to the domain of the weak$^*$ generator; see \cite[Section~3.1]{FG26}.
Since the Hamiltonian is quadratic and the jump operators are linear in the creation and annihilation operators, their images under $\cL$ also have finite matrix support.
Let $\displaystyle \rho=\sum_i s_i|i\rangle\langle i|$, with $s_i>0$, in this basis, and let $Y=Y^*$ have finite matrix support and satisfy $\Tr(\rho Y)=0$.
For $0<p<1$ and $s>t$, define functions $w_p, r_p$ as follows:
\begin{align}
 w_p(s,t)&=\frac1{1-p}\left(t-\frac{s-t}{(s/t)^{1/p}-1}\right),\label{eq:wp}\\
 r_p(s,t)&=\frac{(s-t)\log(s/t)(s/t)^{1/p}}
 {p^2((s/t)^{1/p}-1)^2}.\label{eq:rp}
\end{align}
Extend these kernels symmetrically and put $w_p(s,s)=r_p(s,s)=s$. Denote their Schur multipliers by
\[
 \mathsf W_{p,\rho}(|i\rangle\langle j|)=w_p(s_i,s_j)|i\rangle\langle j|,
 \quad
 \mathsf R_{p,\rho}(|i\rangle\langle j|)=r_p(s_i,s_j)|i\rangle\langle j|.
\]
Only their real quadratic forms are needed; no complete positivity assertion about these multipliers is required.

\begin{lemma}\label{lem:localquotient}
For $Y$ as above and $X_\varepsilon=1+\varepsilon Y$, with $|\varepsilon|\norm Y<1/2$,
\begin{align}
 \log\np{X_\varepsilon}p
 &=\frac{p-1}{2}\varepsilon^2\Tr(Y\mathsf W_{p,\rho}Y)+O(\varepsilon^3),\label{eq:normexp}\\
 \mathcal E_p(X_\varepsilon)
 &=\frac{p^2\varepsilon^2}{4}\Re\Tr(Y\mathsf W_{p,\rho}\cL(Y))+O(\varepsilon^3),\label{eq:energyexp}\\
 \Ent_{p,\rho}(X_\varepsilon)
 &=\frac{p^2\varepsilon^2}{2}\Tr(Y\mathsf R_{p,\rho}Y)+O(\varepsilon^3).\label{eq:entexp}
\end{align}
In particular,
\begin{equation}\label{eq:localquotient}
 \alpha_p(\cL)\leq
 \frac{\Re\Tr(Y\mathsf W_{p,\rho}\cL(Y))}
 {2\Tr(Y\mathsf R_{p,\rho}Y)}
 \quad(Y\ne0).
\end{equation}
\end{lemma}
\begin{proof}
The operator $X_\varepsilon=I+\varepsilon Y$ belongs to $\operatorname{Dom}(\cL)$ and satisfies $X_\varepsilon\geq I/2$.
For $0<p<1$, operator antitonicity of $u\mapsto u^{p-1}$ gives, in the quadratic-form sense,
\[
 0\leq\rho^{1/(2p)}j_p(X_\varepsilon)^{p-1}\rho^{1/(2p)}
 \leq 2^{1-p}\rho.
\]
The right side is trace class. Together with the finite matrix support of $Y$ and $\cL(Y)$, this justifies the trace pairings and differentiations below.
Set $a_i=s_i^{1/p}$ and $K_{ij}=(s_is_j)^{\frac{1}{2p}}Y_{ij}$. The trace Hessian of $A\mapsto\Tr A^p$ at $\diag(a_i)$ is
\[
 p\sum_{i,j}|K_{ij}|^2\frac{a_i^{p-1}-a_j^{p-1}}{a_i-a_j},
\]
with the divided difference replaced by $(p-1)a_i^{p-2}$ on the diagonal. 
Since $\Tr(\rho Y)=0$, the first variation vanishes. 
Rewriting the second variation gives \eqref{eq:normexp}; explicitly its coefficient is $\displaystyle (p-1)\sum_{i,j}w_p(s_i,s_j)|Y_{ij}|^2/2$.

In the time direction, invariance gives $\Tr(\rho\cL (Y))=0$. 
Differentiate the quadratic form in \eqref{eq:normexp}, and use \eqref{eq:tderivative}; this gives \eqref{eq:energyexp}. 
Direct differentiation shows
\[
 r_p(s,t)=w_p(s,t)+(p-1)\partial_pw_p(s,t).
\]
Together with \eqref{eq:pderivative}, this proves \eqref{eq:entexp}. All relevant sums are finite: outside the finite support the input is the identity, while first-order time terms vanish by invariance. Finally divide and let $\varepsilon\to0$. Positivity of $r_p(s,t)$ makes the limiting denominator positive for $Y\ne0$.
\end{proof}

The factor $2$ in the denominator of \eqref{eq:localquotient} is part of the normalization. For a commuting centered direction it gives the familiar linearized upper bound equal to one half of the corresponding Rayleigh quotient.

The two kernel limits are
\begin{align}
 w_{1/2}(s,t)&=\frac{2st}{s+t},\label{eq:halfkernel}\\
 \lim_{p\downarrow0}w_p(s,t)&=\min\{s,t\},\qquad
 \lim_{p\downarrow0}r_p(s,t)=0\quad(s\ne t).\label{eq:zerokernel}
\end{align}
At equal eigenvalues, $r_p(s,s)=s$ does not vanish. Define the spectral-block pinching
\[
 \mathsf P_\rho(Y)=\sum_{s\in\text{Spec}(\rho)\setminus\{0\}}P_sYP_s
\]
and the minimum-kernel multiplier $\mathsf M_\rho(|i\rangle\langle j|)=\min(s_i,s_j)|i\rangle\langle j|$.

\begin{proposition}\label{prop:endpointcriteria}
For a Gaussian quantum Markov semigroup with a faithful invariant Gaussian state, we have that 
\begin{enumerate}[(i)]
\item If finite-support centered self-adjoint $Y_n$ satisfy
\[
 \frac{\Re\Tr(Y_n\mathsf W_{1/2,\rho}\cL(Y_n))}
 {\Tr(Y_n\mathsf R_{1/2,\rho}Y_n)}\longrightarrow0,
\]
then $\alpha_{1/2}(\cL)=0$.
\item If one finite-support self-adjoint $Y\ne0$ satisfies $\mathsf P_\rho(Y)=0$ and
\[
 \Re\Tr(Y\mathsf M_\rho \cL(Y))<0,
\]
then $\displaystyle \lim_{p\downarrow0}\alpha_p(\cL)=-\infty$.
\end{enumerate}
\end{proposition}
\begin{proof}
Part (i) follows from Lemma~\ref{lem:localquotient} and Proposition~\ref{prop:half-nonnegative}. For (ii), the numerator in \eqref{eq:localquotient} converges to a strictly negative number; its denominator is positive and tends to zero by \eqref{eq:zerokernel} and the pinching assumption. The sums are finite, so no exchange of infinite limits is involved.
\end{proof}

The following construction shows that the threshold $1/2$ for nonnegative $p$-Dirichlet forms, and hence for positive infinitesimal reverse hypercontractive rates, is sharp under KMS symmetry, faithful equilibrium, and Hurwitz stability.
For $0<p<1/2$, the negative $p$-Dirichlet forms exclude same-exponent reverse contractivity for all $t\geq0$ and reverse curves with a strictly negative initial exponent derivative.
These conclusions concern behavior at time zero and do not decide whether a fixed pair $0<p<q<1/2$ is attained at sufficiently large times.
The construction also separates these obstructions from nonprimitivity and from an infinite number of modes.


On the Hilbert space $\cF_s(\mathbb C^2)$, take
\begin{equation}\label{eq:endjump}
 L_1=\bfa_1+\bfa_2+\frac12\bfa_1^\dagger+\frac13\bfa_2^\dagger,
\end{equation}
and
\begin{align}
 H={}&\frac{3i}{20}\bigl((\bfa_1^\dagger)^2-\bfa_1^2\bigr)
 +\frac{2i}{15}\bigl((\bfa_2^\dagger)^2-\bfa_2^2\bigr)\notag\\
 &+\frac{25i}{84}\bigl(\bfa_1^\dagger \bfa_2^\dagger-\bfa_1\bfa_2\bigr)
 +\frac{7i}{60}\bigl(\bfa_2^\dagger \bfa_1-\bfa_1^\dagger \bfa_2\bigr).
 \label{eq:endH}
\end{align}
Set
\begin{equation}\label{eq:endgen}
- \cL(X)=i[H,X]+L_1^*XL_1-\frac12\{L_1^*L_1,X\}.
\end{equation}
This is exactly a linear-jump, quadratic-Hamiltonian generator of the form \eqref{eq:H}--\eqref{eq:GKLS}.
Existence of its conservative Gaussian quantum Markov semigroup follows from \cite[Theorem 4]{FG26}. 
On finite-particle vectors all coefficient identities below are ordinary polynomial identities.

Its faithful invariant density, as verified below, is
\begin{equation}\label{eq:endrho}
 \rho=\rho_{\log4}\otimes\rho_{\log9}
 =\frac23\,4^{-\bfa_1^\dag\bfa_1}9^{-\bfa_2^\dag\bfa_2}.
\end{equation}

In the real-coordinate order $(\Re z_1,\Re z_2,\Im z_1,\Im z_2)$, the matrices are
\begin{equation}\label{eq:endZ}
 \bfZ=\begin{pmatrix}
 -27/40&-18/35&0&0\\
 -32/35&-32/45&0&0\\
 0&0&-3/40&-3/35\\
 0&0&-16/105&-8/45
 \end{pmatrix},
\end{equation}
\begin{equation}\label{eq:endC}
 \bfC=\begin{pmatrix}
 9/4&2&0&0\\2&16/9&0&0\\
 0&0&1/4&1/3\\0&0&1/3&4/9
 \end{pmatrix},
 \qquad
 \Sigma=\diag(5/3,5/4,5/3,5/4).
\end{equation}
Substitution into \eqref{eq:ZUV}--\eqref{eq:CUV} gives these matrices. 

\begin{lemma}\label{lem:endstructure}
We have that 
\begin{gather}
 \bfZ^*\Sigma+\Sigma\bfZ+\bfC=0,\qquad \Sigma+i\bfJ>0,\label{eq:endstationary}\\
 \bfC_{\bfZ}=\mathbf v\mathbf v^*,\qquad
 \mathbf v=(3/2,4/3,-i/2,-2i/3)^{\mathsf{T}},\label{eq:endCZ}\\
 \Sigma_{\KMS}=\diag(4/3,3/4,4/3,3/4),\qquad
 \Sigma_{\KMS}\bfZ=\bfZ^*\Sigma_{\KMS}.\label{eq:endSK}
\end{gather}
The drift is Hurwitz and the semigroup is primitive and KMS-symmetric.
\end{lemma}
\begin{proof}
The Lyapunov and rank-one identities follow by rational matrix multiplication. The symplectic eigenvalues of $\Sigma$ are $5/3$ and $5/4$, so its faithful density is the product with inverse temperatures $\log4$ and $\log9$, namely \eqref{eq:endrho}. 
Formula \eqref{eq:SK} gives \eqref{eq:endSK}; Proposition~\ref{prop:kmssym} proves KMS symmetry.

Write $Z_P$ and $Z_Q$ for the upper-left and lower-right $2\times2$ blocks of \eqref{eq:endZ}. 
Their traces and determinants are
\begin{align*}
 \Tr (Z_P)=-499/360,&\qquad \det (Z_P)=12/1225,\\
 \Tr (Z_Q)=-91/360,&\qquad \det( Z_Q)=1/3675.
\end{align*}
Both real blocks are Hurwitz. 
It follows from Proposition~\ref{prop:hurwitzpos} that $\bfC_{\KMS}>0$. Also
\[
 \bfF=\int_0^\infty e^{s\bfZ^*}\bfC_{\bfZ} e^{s\bfZ}\,ds>0,
\]
so $(\bfC_{\bfZ}\text{-cond})$ holds by Proposition~\ref{prop:control}. Irreducibility follows from \cite[Theorem 12]{FG26}. For completeness, invariance and the Weyl formula also prove uniqueness directly: if $\omega$ is any invariant normal state, then
\[
 \Tr(\omega W(\mathbf z))
 =e^{-\mathbf z^{\mathsf{T}}\bfC_t\mathbf z/2}\Tr\left(\omega W(e^{t\bfZ}\mathbf z)\right)
 \longrightarrow e^{-\mathbf z^{\mathsf{T}}\Sigma\mathbf z/2}.
\]
Thus $\omega=\rho$. The invariant state is faithful, and the fixed algebra is scalar (alternatively this follows at once from the strictly contracting nonvacuum KMS chaoses). 
\end{proof}

The local test will be the field
\begin{equation}\label{eq:endY}
 Y=(\bfa_2+\bfa_2^\dagger)-(\bfa_1+\bfa_1^\dagger)=R(\mathbf z),
 \qquad \mathbf z=(0,0,-1,1)^{\mathsf{T}}.
\end{equation}
It is unbounded; we do not assert positivity of $1+\varepsilon Y$. 

For a one-mode thermal density, summing the kernels \eqref{eq:wp}--\eqref{eq:rp} on a ladder direction gives
\begin{align}
 \mathsf w_p(\beta)
 &=\frac{2}{1-p}\left(\frac1{e^\beta-1}-\frac1{e^{\beta/p}-1}\right),
 \label{eq:fieldw}\\
 \mathsf r_p(\beta)
 &=\frac{\beta}{2p^2\sinh^2(\beta/(2p))}.
 \label{eq:fieldr}
\end{align}
Indeed, the neighboring eigenvalues have ratio $e^\beta$, and the two symmetric ladder entries contribute the sum
\[
 2\sum_{n\geq0}(n+1)\,w_p((1-e^{-\beta})e^{-n\beta},(1-e^{-\beta})e^{-(n+1)\beta}),
\]
which evaluates to \eqref{eq:fieldw};
the calculation for $r_p$ is identical. 
Other modes sum to one. 
Entries associated with distinct ladder directions have disjoint matrix support, 
so the resulting first-chaos forms are diagonal in the mode index. 
The position and momentum directions have the same weights.


Let
\[
 P_N=\sum_{0\leq n_1,n_2\leq N}
 |n_1,n_2\rangle\langle n_1,n_2|,
 \qquad Y_N=P_NYP_N.
\]
The operator $Y_N$ is bounded and self-adjoint with finite matrix support.
It also satisfies $\Tr(\rho Y_N)=0$.
For $0<|\varepsilon|<(2\norm{Y_N})^{-1}$,
\begin{equation}\label{eq:boundedwitness}
 X_{N,\varepsilon}=I+\varepsilon Y_N\geq\frac12 I.
\end{equation}
Since $Y_N$ is a Schwartz operator,
$Y_N\in\operatorname{Dom}(\cL)$; see \cite{FG26}.
Moreover, $\cL(Y_N)$ has finite matrix support and $\cL(I)=0$.
Thus the trace differentiations used in
Lemma~\ref{lem:localquotient} are justified for
$X_{N,\varepsilon}$.

\begin{lemma}\label{lem:cutoff}
For every fixed $0<p<1$,
\begin{align}
 \Re\Tr(Y_N\mathsf W_{p,\rho}\cL (Y_N))
 &\longrightarrow-\frac3{280}\mathsf w_p(\log4)
                    +\frac8{315}\mathsf w_p(\log9),\label{eq:cutoffnum}\\
 \Tr(Y_N\mathsf R_{p,\rho}Y_N)
 &\longrightarrow\mathsf r_p(\log4)+\mathsf r_p(\log9).
 \label{eq:cutoffden}
\end{align}
Moreover, $\Re\Tr(Y_N\mathsf M_\rho \cL(Y_N))\to-1/1260$.
\end{lemma}
\begin{proof}
The coefficients of $Y$ and $\cL(Y)$ grow at most polynomially in the occupation numbers, and all nonzero entries of $Y_N$ join neighboring occupations of one mode. 
Their eigenvalue ratios are exactly $4$ or $9$. The formulas \eqref{eq:wp}--\eqref{eq:rp} therefore bound their weights by $C_p(s_i+s_j)$, with $C_p$ independent of $N$.

The GKLS expression has fixed finite bandwidth in the occupation basis.
Away from a boundary strip of fixed width around $n_j=N$, the coefficients of $\cL(Y_N)$ and $\cL(Y)$ agree. 
On the boundary strip their difference has polynomial growth: every term contains at most four ladder factors in the product with $Y_N$. With thermal weights $\displaystyle s_{n_1,n_2}=\frac{2}{3}4^{-n_1}9^{-n_2}$, the omitted and boundary contributions are bounded by a constant times
\begin{align*}
 \sum_{n_1\geq N-c\ \text{or}\ n_2\geq N-c}
 (1+n_1+n_2)^k4^{-n_1}9^{-n_2},
\end{align*}
for fixed integers $c,k$, and hence tend to zero. 
The bulk sums are precisely the ladder sums giving \eqref{eq:fieldw}--\eqref{eq:fieldr}. 
The same proof applies to the minimum kernel because $\min(s_i,s_j)\leq s_i+s_j$.
\end{proof}

Let
\begin{equation}\label{eq:Up}
\delta(p)=
 \frac{-\frac3{280}\mathsf w_p(\log4)+\frac8{315}\mathsf w_p(\log9)}
 {2\big(\mathsf r_p(\log4)+\mathsf r_p(\log9)\big)}.
\end{equation}

For fixed $0<p<1$ and $N$, apply Lemma~\ref{lem:localquotient} to $Y_N$, or equivalently let $\varepsilon\to0$ in the quotient for \eqref{eq:boundedwitness}.
Then let $N\to\infty$ and use Lemma~\ref{lem:cutoff} to obtain
\begin{equation}\label{eq:alphaUp}
 \alpha_p(\cL)\leq \delta(p),\qquad 0<p<1.
\end{equation}

At $p=1/2$,
\[
 \mathsf w_{1/2}(\log4)=16/15,\qquad
 \mathsf w_{1/2}(\log9)=9/20.
\]
The numerator of \eqref{eq:Up} is zero.
The denominator is strictly positive. 
Thus \eqref{eq:alphaUp} gives $\alpha_{1/2}(\cL)\leq0$, and the endpoint nonnegativity in Proposition~\ref{prop:half-nonnegative} proves
\begin{align}\label{eq:alphahalf}
\alpha_{1/2}(\cL)=0.
\end{align}
The relevant matrix degeneracy is explicit. Put
\[
 W_{1/2}=\diag(16/15,9/20).
\]
Then
\begin{equation}\label{eq:halfmatrix}
 -\tfrac12(Z_Q^{\mathsf{T}}W_{1/2}+W_{1/2}Z_Q)
 =\frac2{25}\begin{pmatrix}1&1\\1&1\end{pmatrix}.
\end{equation}
Its kernel contains $(-1,1)^{\mathsf{T}}$. This, not singularity of $\bfC_{\bfZ}$ by itself, proves the half-order obstruction. The infimum is approached by bounded tests; it need not be attained by a single bounded nonscalar observable.

For $0<p<1/2$, substituting \eqref{eq:fieldw} rewrites the numerator of \eqref{eq:Up} as
\[
 \frac{2}{2520(1-p)}
 \left(-1+\frac{27}{4^{1/p}-1}-\frac{64}{9^{1/p}-1}\right).
\]
Put $f(x)=27/(4^x-1)-64/(9^x-1)$, so that $f(2)=1$.
For $x\geq2$, the ratio of the positive term in $f'(x)$ to the absolute value of its negative term satisfies
\[
 \frac{64\log9}{27\log4}\left(\frac49\right)^x
 \left(\frac{1-4^{-x}}{1-9^{-x}}\right)^2
 <\frac{128}{27}\left(\frac49\right)^2
 =\frac{2048}{2187}<1.
\]
Here we used $\log9<2\log4$ and $(1-4^{-x})/(1-9^{-x})<1$.
Thus $f'(x)<0$ for $x\geq2$, and $f(1/p)<1$ whenever $0<p<1/2$.
The denominator of \eqref{eq:Up} is strictly positive, so \eqref{eq:alphaUp} gives
\[
 \alpha_p(\cL)\leq\delta(p)<0,\qquad 0<p<1/2.
\]

As $p\downarrow0$,
\[
 \mathsf w_p(\log4)\longrightarrow2/3,
 \qquad \mathsf w_p(\log9)\longrightarrow1/4.
\]
Hence the numerator tends to the exact negative number
\begin{equation}\label{eq:negativeconstant}
 -\frac3{280}\frac23+\frac8{315}\frac14=-\frac1{1260}.
\end{equation}
Meanwhile
\[
 \mathsf r_p(\beta)\sim\frac{2\beta}{p^2}e^{-\beta/p}.
\]
The smaller inverse temperature $\log4$ controls the denominator, and
\begin{equation}\label{eq:Uasymp}
\delta(p)\sim-\frac{p^2\,4^{1/p}}{5040\log4}\longrightarrow-\infty.
\end{equation}
Thus \eqref{eq:alphaUp} implies
\begin{align}\label{eq:alphainfty}
  \lim_{p\downarrow0}\alpha_p(\cL)=-\infty.  
\end{align}

On the same two-dimensional block, the limiting matrix is
\begin{equation}\label{eq:zeromatrix}
 W_0=\diag(2/3,1/4),\qquad
 -\tfrac12(Z_Q^{\mathsf{T}}W_0+W_0Z_Q)
 =\begin{pmatrix}1/20&1/21\\1/21&2/45\end{pmatrix}.
\end{equation}
Its quadratic form on $(-1,1)^{\mathsf{T}}$ is $-1/1260$. The half-order null direction has become a negative direction in the limiting weighted energy, while the entropy Hessian vanishes.

A fixed exponent already gives a strict obstruction. At $p=1/4$,
\[
 \mathsf w_{1/4}(\log4)=224/255,\qquad
 \mathsf w_{1/4}(\log9)=273/820,
\]
and the numerator is $-2/2091$. Therefore
\begin{equation}\label{eq:quarterbound}
 \alpha_{1/4}(\cL)\leq
 -\left[2091\left(\frac{8192\log4}{65025}
                  +\frac{6561\log9}{1344800}\right)\right]^{-1}
 \simeq-0.0025799443.
\end{equation}
Some values of the \emph{upper bound} $\delta(p)$ are recorded in \Cref{tab:endpoint}.
They are not numerically computed optimal logarithmic Sobolev constants.
\begin{table}[htbp]
\centering
\caption{Local upper bounds for the primitive two-mode example.}\label{tab:endpoint}
\begin{tabular}{@{}rr@{}}
\hline
$p$ & $\delta(p)$\\ 
\hline
$0.50$ & $0$\\
$0.49$ & $-2.8401242\times10^{-5}$\\
$0.40$ & $-3.7663422\times10^{-4}$\\
$0.25$ & $-2.5799443\times10^{-3}$\\
$0.10$ & $-1.6666809$\\
$0.05$ & $-4.1412378\times10^5$\\ 
\hline
\end{tabular}
\par\smallskip
{\footnotesize At $p=1/2$, the equality $\alpha_{1/2}(\cL)=0$ is proved in \eqref{eq:alphahalf}.\par}
\end{table}



Since the drift is Hurwitz, the whole system is the relaxing sector in Theorem~\ref{thm:section4-relative-cmlsi}.
The rank-one identity \eqref{eq:endCZ} and $\mathbf F>0$ therefore give
\[
 \alpha_{\mathrm{CMLSI}}(\Phi)
 =\lambda_{\min}\!\left(\mathbf F^{-1/2}\bfC_{\bfZ}\mathbf F^{-1/2}\right)=0.
\]
Together with the preceding estimates, this proves the following theorem.
\begin{theorem}\label{thm:endpoint}
There is a two-mode primitive KMS-symmetric Gaussian quantum Markov semigroup with faithful invariant density \eqref{eq:endrho} for which
\begin{equation}\label{eq:endpoints-main}
 \begin{gathered}
 \alpha_{\mathrm{CMLSI}}(\Phi)=0,\qquad \alpha_{1/2}(\cL)=0,\\
 \alpha_p(\cL)<0\quad(0<p<1/2),\qquad
 \lim_{p\downarrow0}\alpha_p(\cL)=-\infty.
 \end{gathered}
\end{equation}
Its drift is Hurwitz, its KMS-dissipation matrix is strictly positive, and its quantum-diffusion matrix has rank one. 
Thus it has all the positive forward and above-$1/2$ reverse properties in Theorems \ref{thm:forward} and \ref{thm:reverse}, while satisfying Equations \eqref{eq:endpoints-main}.
\end{theorem}

\section*{Declaration of generative AI use}

During the preparation of this work, the authors used generative AI
tools to improve the language and readability of the manuscript and
to assist with developing and checking mathematical arguments.
The authors reviewed and revised the AI-assisted content and verified
the mathematical statements and proofs.
The authors take full responsibility for the accuracy and integrity
of the manuscript.

\bibliographystyle{plain}
\bibliography{hypergau}
\end{document}